%% file: main.tex
\documentclass[letterpaper, 10 pt, conference]{ieeeconf}

\IEEEoverridecommandlockouts%
\usepackage{cite}
\usepackage{amsmath}
\usepackage{amssymb}
\usepackage{amsfonts}
\usepackage{algorithmic}
\usepackage{graphicx}
\usepackage{textcomp}
\usepackage{xcolor}

\def\BibTeX{{\rm B\kern-.05em{\sc i\kern-.025em b}\kern-.08em
    T\kern-.1667em\lower.7ex\hbox{E}\kern-.125emX}}

\input{my_packages}

\input{my_commands}

\input{Styles/colors}
\input{Styles/environments}
\input{Styles/graphs}
\input{Styles/spacing}

\begin{document}

\title{Quantifying Control Performance Loss for a\\ Least Significant Bits Authentication Scheme}

% 1: IEEEtran.cls
% \author{\IEEEauthorblockN{1\textsuperscript{st} Given Name Surname}
% \IEEEauthorblockA{\textit{dept. name of organization (of Aff.)} \\
% \textit{name of organization (of Aff.)}\\
% City, Country \\
% email address or ORCID}
% \and
% \IEEEauthorblockN{2\textsuperscript{nd} Given Name Surname}
% \IEEEauthorblockA{\textit{dept. name of organization (of Aff.)} \\
% \textit{name of organization (of Aff.)}\\
% City, Country \\
% email address or ORCID}
% \and
% \IEEEauthorblockN{3\textsuperscript{rd} Given Name Surname}
% \IEEEauthorblockA{\textit{dept. name of organization (of Aff.)} \\
% \textit{name of organization (of Aff.)}\\
% City, Country \\
% email address or ORCID}
% \and
% \IEEEauthorblockN{4\textsuperscript{th} Given Name Surname}
% \IEEEauthorblockA{\textit{dept. name of organization (of Aff.)} \\
% \textit{name of organization (of Aff.)}\\
% City, Country \\
% email address or ORCID}
% \and
% \IEEEauthorblockN{5\textsuperscript{th} Given Name Surname}
% \IEEEauthorblockA{\textit{dept. name of organization (of Aff.)} \\
% \textit{name of organization (of Aff.)}\\
% City, Country \\
% email address or ORCID}
% \and
% \IEEEauthorblockN{6\textsuperscript{th} Given Name Surname}
% \IEEEauthorblockA{\textit{dept. name of organization (of Aff.)} \\
% \textit{name of organization (of Aff.)}\\
% City, Country \\
% email address or ORCID}
% }
% 2: CSSconf.cls
% FIXME: This thanks{} is too far indented..
\author{Bart Wolleswinkel and Riccardo Ferrari%
	\thanks{
		This work has been partially supported by the EU Horizon program through the project TWAIN, grant id 101122194.}
	\thanks{The authors are with the Delft Center for Systems and Control (DCSC), Mechanical Engineering Faculty, Delft University of Technology, Delft, The Netherlands (email:\{\texttt{b.wolleswinkel},\texttt{r.ferrari}\}\texttt{@tudelft.nl}).
	}
}

\maketitle%

\begin{abstract}
    \Acfp{ICS} often consist of many legacy devices, which were designed without security requirements in mind. With the increase in cyberattacks targeting critical infrastructure, there is a growing urgency to develop legacy-compatible security solutions tailored to the specific needs and constraints of real-time control systems. We propose a~\acfp{LSB} coding scheme providing message authentication and integrity, which is compatible with legacy devices and never compromises availability. The scheme comes with provable security guarantees, and we provide a simple yet effective method to deal with synchronization issues due to packet dropouts. Furthermore, we quantify the control performance loss for both a fixed-point and floating-point quantization architecture when using the proposed coding scheme. We demonstrate its effectiveness in detecting cyberattacks, as well as the impact on control performance, on a hydro power turbine control system.
    %
    %  by relaxing the constraint that the measurement as induced by the attack needs to be identically zero
    %
    % By doing so, the adversary can take saturation constraints on the actuators, stealthiness of the measured output, and disruptiveness of the final state into consideration.
\end{abstract}

% Include main matter
\input{Sections/ecc_2026}

% Include acronyms
\input{Sections/acronyms}

% ------------ BIBLIOGRAPHY ------------

\bibliographystyle{IEEEtran}
\bibliography{references}

% ==== DRAFT ====
% \clearpage%
% \input{Sections/reviews}
% ==== DRAFT ====

\end{document}

%% file: my_packages.tex
\usepackage{amsmath}
\usepackage{amssymb}
\usepackage{amsthm}  % For better theorems
\usepackage{mathtools}
\usepackage{accents} % For tilde symbol under a letters
\usepackage{scalerel} % For stretching symbols
\usepackage{etoolbox} % For Arabic numerals in sub-equations
\usepackage{siunitx}  % For better unit handling
\usepackage[normalem]{ulem}  % Underlining
\patchcmd{\subequations}{\def\theequation{\theparentequation\alph{equation}}}%
{\def\theequation{\theparentequation.\arabic{equation}}}{}{}

\usepackage{hyperref}
\hypersetup{
    colorlinks=true,  % If false, text will be black but boxes will appear
    linkcolor=black,  % For changing \ref{} AND \ac{}
    urlcolor=mydarkblue,  % For changing \href{}{}
    citecolor=mydarkblue,  % For changing \cite{}
    pdftitle={Your title: go to packages > \textbackslash hyperref}
}

\usepackage{cleveref}

\usepackage[nolist]{acronym}

\usepackage{booktabs}
\usepackage{multirow}
\usepackage{multicol}
\usepackage{array}
    \newcolumntype{L}{>{$}l<{$}}
    \newcolumntype{C}{>{$}c<{$}}
    \newcolumntype{R}{>{$}r<{$}}

\usepackage{enumerate}

\usepackage{xargs}
\usepackage{ifthen} % Provides \ifthenelse test  
\usepackage{xifthen} % Provides \isempty test
\usepackage{lettrine}  % For large first letter
\usepackage{bm}  % For better bold math

\usepackage{tikz}
\usetikzlibrary{positioning, backgrounds} % For relative positioning and drawing things on a background
\usepackage{pgfplots} 
\usepackage{pgfgantt}
\usepackage{pdflscape}
\pgfplotsset{compat=newest} 
\pgfplotsset{plot coordinates/math parser=false}
\usepgfplotslibrary{fillbetween}
\usetikzlibrary{positioning}

\usepackage[hang]{footmisc} % For better spacing
\usepackage[framemethod=TikZ]{mdframed} % For better coloured frames

%% file: my_commands.tex
\newcommand{\R}{\mathbb{R}}

\newcommand{\Real}[2][]{%
    {\mathrm{Re} \mspace{-2mu}}
    \ifthenelse{\isempty{#1}}
        {\left\{ \mspace{1mu} #2 \mspace{1mu} \right\}} % if no number is given
    {\ifthenelse{\equal{#1}{*}}
        {\left\{ \mspace{1mu} #2 \mspace{1mu} \right\}} % if * is given
    {\ifthenelse{\equal{#1}{0}}
        {\{ \mspace{1mu} #2 \mspace{1mu} \}} % if 0 is given
    {\ifthenelse{\equal{#1}{1}}
        {\big\{ \mspace{1mu} #2 \mspace{1mu} \big\}} % if 1 is given
    {\ifthenelse{\equal{#1}{2}}
        {\Big\{ \mspace{1mu} #2 \mspace{1mu} \Big\}} % if 2 is given
    {\ifthenelse{\equal{#1}{3}}
        {\bigg\{ \mspace{1mu} #2 \mspace{1mu} \bigg\}} % if 3 is given
    {\ifthenelse{\equal{#1}{4}}
        {\Bigg\{ \mspace{1mu} #2 \mspace{1mu} \Bigg\}} % if 4 is given
    {\PackageError{mycommands}{#1 is an invalid optional argument}{Values should be between 0 and 4}}}}}}}} % else
}

\newcommand{\Imag}[2][]{%
    {\mathrm{Im} \mspace{-2mu}}
    \ifthenelse{\isempty{#1}}
        {\left\{ \mspace{1mu} #2 \mspace{1mu} \right\}} % if no number is given
    {\ifthenelse{\equal{#1}{*}}
        {\left\{ \mspace{1mu} #2 \mspace{1mu} \right\}} % if * is given
    {\ifthenelse{\equal{#1}{0}}
        {\{ \mspace{1mu} #2 \mspace{1mu} \}} % if 0 is given
    {\ifthenelse{\equal{#1}{1}}
        {\big\{ \mspace{1mu} #2 \mspace{1mu} \big\}} % if 1 is given
    {\ifthenelse{\equal{#1}{2}}
        {\Big\{ \mspace{1mu} #2 \mspace{1mu} \Big\}} % if 2 is given
    {\ifthenelse{\equal{#1}{3}}
        {\bigg\{ \mspace{1mu} #2 \mspace{1mu} \bigg\}} % if 3 is given
    {\ifthenelse{\equal{#1}{4}}
        {\Bigg\{ \mspace{1mu} #2 \mspace{1mu} \Bigg\}} % if 4 is given
    {\PackageError{mycommands}{#1 is an invalid optional argument}{Values should be between 0 and 4}}}}}}}} % else
}

\newcommand\vctr[1]{%
    \ifcat\noexpand#1\relax % Check if the argument is a control sequence
        \bm{#1} % Probably Greek
    \else
        \bm{#1} % Single character
    \fi
} % Change to either boldface or boldface + italic

\newcommand{\mtrx}[1]{%
    \ifcat\noexpand#1\relax % Check if the argument is a control sequence
        \bm{#1} % Probably Greek
    \else
        \bm{#1} % Single character
    \fi
} % Change to either boldface or boldface + italic

\newcommand{\T}{\mathsf{T}}

\newcommand{\ccdot}{\mathchoice
  {\hspace{-0.15em}\cdot\hspace{-0.15em}}%
  {\hspace{-0.15em}\cdot\hspace{-0.15em}}%
  {\cdot}%
  {\cdot}%
}

\newcommand{\vnorm}[2][]{%
    \ifthenelse{\isempty{#1}}
        {\left\Vert \mspace{1mu} #2 \mspace{1mu} \right\Vert} % if no number is given
    {\ifthenelse{\equal{#1}{*}}
        {\left\Vert \mspace{1mu} #2 \mspace{1mu} \right\Vert} % if * is given
    {\ifthenelse{\equal{#1}{0}}
        {\Vert \mspace{1mu} #2 \mspace{1mu} \Vert} % if 0 is given
    {\ifthenelse{\equal{#1}{1}}
        {\big\Vert \mspace{1mu} #2 \mspace{1mu} \big\Vert} % if 1 is given
    {\ifthenelse{\equal{#1}{2}}
        {\Big\Vert \mspace{1mu} #2 \mspace{1mu} \Big\Vert} % if 2 is given
    {\ifthenelse{\equal{#1}{3}}
        {\bigg\Vert \mspace{1mu} #2 \mspace{1mu} \bigg\Vert} % if 3 is given
    {\ifthenelse{\equal{#1}{4}}
        {\Bigg\Vert \mspace{1mu} #2 \mspace{1mu} \Bigg\Vert} % if 4 is given
    {\PackageError{mycommands}{#1 is an invalid optional argument}{Values should be between 0 and 4}}}}}}}} % else
}

\newcommand{\abs}[2][]{%
    \ifthenelse{\isempty{#1}}
        {\left\vert \mspace{1mu} #2 \mspace{1mu} \right\vert} % if no number is given
    {\ifthenelse{\equal{#1}{*}}
        {\left\vert \mspace{1mu} #2 \mspace{1mu} \right\vert} % if * is given
    {\ifthenelse{\equal{#1}{0}}
        {\vert \mspace{1mu} #2 \mspace{1mu} \vert} % if 0 is given
    {\ifthenelse{\equal{#1}{1}}
        {\big\vert \mspace{1mu} #2 \mspace{1mu} \big\vert} % if 1 is given
    {\ifthenelse{\equal{#1}{2}}
        {\Big\vert \mspace{1mu} #2 \mspace{1mu} \Big\vert} % if 2 is given
    {\ifthenelse{\equal{#1}{3}}
        {\bigg\vert \mspace{1mu} #2 \mspace{1mu} \bigg\vert} % if 3 is given
    {\ifthenelse{\equal{#1}{4}}
        {\Bigg\vert \mspace{1mu} #2 \mspace{1mu} \Bigg\vert} % if 4 is given
    {\PackageError{mycommands}{#1 is an invalid optional argument}{Values should be between 0 and 4}}}}}}}} % else
}

\DeclareMathOperator*{\argmin}{arg\,min}
\newcommand{\Exp}[2][]{%
    \ifthenelse{\isempty{#1}}
        {\mathbb{E} \mspace{-2mu} \left[ \mspace{1mu} #2 \mspace{1mu} \right]} % if no number is given
    {\ifthenelse{\equal{#1}{*}}
        {\mathbb{E} \mspace{-1mu} \left[ \mspace{1mu} #2 \mspace{1mu} \right]} % if * is given
    {\ifthenelse{\equal{#1}{0}}
        {\mathbb{E} [ \mspace{1mu} #2 \mspace{1mu} ]} % if 0 is given
    {\ifthenelse{\equal{#1}{1}}
        {\mathbb{E} \big[ \mspace{1mu} #2 \mspace{1mu} \big]} % if 1 is given
    {\ifthenelse{\equal{#1}{2}}
        {\mathbb{E} \Big[ \mspace{1mu} #2 \mspace{1mu} \Big]} % if 2 is given
    {\ifthenelse{\equal{#1}{3}}
        {\mathbb{E} \bigg[ \mspace{1mu} #2 \mspace{1mu} \bigg]} % if 3 is given
    {\ifthenelse{\equal{#1}{4}}
        {\mathbb{E} \Bigg[ \mspace{1mu} #2 \mspace{1mu} \Bigg]} % if 4 is given
    {\PackageError{mycommands}{#1 is an invalid optional argument}{Values should be between 0 and 4}}}}}}}} % else
}

\newcommand{\Var}[2][]{%
    \ifthenelse{\isempty{#1}}
        {\mathbb{V}\mathrm{ar} \mspace{-2mu} \left[ \mspace{1mu} #2 \mspace{1mu} \right]} % if no number is given
    {\ifthenelse{\equal{#1}{*}}
        {\mathbb{V}\mathrm{ar} \mspace{-1mu} \left[ \mspace{1mu} #2 \mspace{1mu} \right]} % if * is given
    {\ifthenelse{\equal{#1}{0}}
        {\mathbb{V}\mathrm{ar} [ \mspace{1mu} #2 \mspace{1mu} ]} % if 0 is given
    {\ifthenelse{\equal{#1}{1}}
        {\mathbb{V}\mathrm{ar} \big[ \mspace{1mu} #2 \mspace{1mu} \big]} % if 1 is given
    {\ifthenelse{\equal{#1}{2}}
        {\mathbb{V}\mathrm{ar} \Big[ \mspace{1mu} #2 \mspace{1mu} \Big]} % if 2 is given
    {\ifthenelse{\equal{#1}{3}}
        {\mathbb{V}\mathrm{ar} \bigg[ \mspace{1mu} #2 \mspace{1mu} \bigg]} % if 3 is given
    {\ifthenelse{\equal{#1}{4}}
        {\mathbb{V}\mathrm{ar} \Bigg[ \mspace{1mu} #2 \mspace{1mu} \Bigg]} % if 4 is given
    {\PackageError{mycommands}{#1 is an invalid optional argument}{Values should be between 0 and 4}}}}}}}} % else
}

\newcommand{\Cov}[2][]{%
    \ifthenelse{\isempty{#1}}
        {\mathbb{C}\mathrm{ov} \mspace{-2mu} \left[ \mspace{1mu} #2 \mspace{1mu} \right]} % if no number is given
    {\ifthenelse{\equal{#1}{*}}
        {\mathbb{C}\mathrm{ov} \mspace{-1mu} \left[ \mspace{1mu} #2 \mspace{1mu} \right]} % if * is given
    {\ifthenelse{\equal{#1}{0}}
        {\mathbb{C}\mathrm{ov} [ \mspace{1mu} #2 \mspace{1mu} ]} % if 0 is given
    {\ifthenelse{\equal{#1}{1}}
        {\mathbb{C}\mathrm{ov} \big[ \mspace{1mu} #2 \mspace{1mu} \big]} % if 1 is given
    {\ifthenelse{\equal{#1}{2}}
        {\mathbb{C}\mathrm{ov} \Big[ \mspace{1mu} #2 \mspace{1mu} \Big]} % if 2 is given
    {\ifthenelse{\equal{#1}{3}}
        {\mathbb{C}\mathrm{ov} \bigg[ \mspace{1mu} #2 \mspace{1mu} \bigg]} % if 3 is given
    {\ifthenelse{\equal{#1}{4}}
        {\mathbb{C}\mathrm{ov} \Bigg[ \mspace{1mu} #2 \mspace{1mu} \Bigg]} % if 4 is given
    {\PackageError{mycommands}{#1 is an invalid optional argument}{Values should be between 0 and 4}}}}}}}} % else
}

\DeclareMathSymbol{\shortminus}{\mathbin}{AMSa}{"39}

\newcommand{\set}[2][]{%
    \ifthenelse{\isempty{#1}}
        {\left\{ \mspace{1mu} #2 \mspace{1mu} \right\}} % if no number is given
    {\ifthenelse{\equal{#1}{*}}
        {\left\{ \mspace{1mu} #2 \mspace{1mu} \right\}} % if * is given
    {\ifthenelse{\equal{#1}{0}}
        {\{ \mspace{1mu} #2 \mspace{1mu} \}} % if 0 is given
    {\ifthenelse{\equal{#1}{1}}
        {\big\{ \mspace{1mu} #2 \mspace{1mu} \big\}} % if 1 is given
    {\ifthenelse{\equal{#1}{2}}
        {\Big\{ \mspace{1mu} #2 \mspace{1mu} \Big\}} % if 2 is given
    {\ifthenelse{\equal{#1}{3}}
        {\bigg\{ \mspace{1mu} #2 \mspace{1mu} \bigg\}} % if 3 is given
    {\ifthenelse{\equal{#1}{4}}
        {\Bigg\{ \mspace{1mu} #2 \mspace{1mu} \Bigg\}} % if 4 is given
    {\PackageError{mycommands}{#1 is an invalid optional argument}{Values should be between 0 and 4}}}}}}}} % else
}

\newcommand{\forwhich}[1][]{%
    \ifthenelse{\isempty{#1}}
        {\mspace{3mu} \vert \mspace{3mu}} % if no number is given
    {\ifthenelse{\equal{#1}{*}}
        {\mspace{3mu} \vert \mspace{3mu}} % if * is given
    {\ifthenelse{\equal{#1}{0}}
        {\mspace{3mu} \vert \mspace{3mu}} % if 0 is given
    {\ifthenelse{\equal{#1}{1}}
        {\mspace{3mu} \big\vert \mspace{3mu}} % if 1 is given
    {\ifthenelse{\equal{#1}{2}}
        {\mspace{3mu} \Big\vert \mspace{3mu}} % if 2 is given
    {\ifthenelse{\equal{#1}{3}}
        {\mspace{3mu} \bigg\vert \mspace{3mu}} % if 3 is given
    {\ifthenelse{\equal{#1}{4} }
        {\mspace{3mu} \Bigg\vert \mspace{3mu}} % if 4 is given
    {\PackageError{mycommands}{#1 is an invalid optional argument}{Values should be between 0 and 4}}}}}}}} % else
}

\newcommand{\given}[1][]{%
    \ifthenelse{\isempty{#1}}
        {\mspace{3mu} \vert \mspace{3mu}} % if no number is given
    {\ifthenelse{\equal{#1}{*}}
        {\mspace{3mu} \vert \mspace{3mu}} % if * is given
    {\ifthenelse{\equal{#1}{0}}
        {\mspace{3mu} \vert \mspace{3mu}} % if 0 is given
    {\ifthenelse{\equal{#1}{1}}
        {\mspace{3mu} \big\vert \mspace{3mu}} % if 1 is given
    {\ifthenelse{\equal{#1}{2}}
        {\mspace{3mu} \Big\vert \mspace{3mu}} % if 2 is given
    {\ifthenelse{\equal{#1}{3}}
        {\mspace{3mu} \bigg\vert \mspace{3mu}} % if 3 is given
    {\ifthenelse{\equal{#1}{4} }
        {\mspace{3mu} \Bigg\vert \mspace{3mu}} % if 4 is given
    {\PackageError{mycommands}{#1 is an invalid optional argument}{Values should be between 0 and 4}}}}}}}} % else
}

\newcommand{\diag}[2][]{%
    {\mathrm{diag} \mspace{-2mu}}
    \ifthenelse{\isempty{#1}}
        {\left( \mspace{1mu} #2 \mspace{1mu} \right)} % if no number is given
    {\ifthenelse{\equal{#1}{*}}
        {\left( \mspace{1mu} #2 \mspace{1mu} \right)} % if * is given
    {\ifthenelse{\equal{#1}{0}}
        {(\mspace{1mu} #2 \mspace{1mu})} % if 0 is given
    {\ifthenelse{\equal{#1}{1}}
        {\big( \mspace{1mu} #2 \mspace{1mu} \big)} % if 1 is given
    {\ifthenelse{\equal{#1}{2}}
        {\Big( \mspace{1mu} #2 \mspace{1mu} \Big)} % if 2 is given
    {\ifthenelse{\equal{#1}{3}}
        {\bigg( \mspace{1mu} #2 \mspace{1mu} \bigg)} % if 3 is given
    {\ifthenelse{\equal{#1}{4}}
        {\Bigg( \mspace{1mu} #2 \mspace{1mu} \Bigg)} % if 4 is given
    {\PackageError{mycommands}{#1 is an invalid optional argument}{Values should be between 0 and 4}}}}}}}} % else
}

\newcommand{\Stuxnet}{\textup{S}\textsc{tuxnet}}

\newcommand{\rndm}[2][]{%
    \ifthenelse{\isempty{#1}}
        {\underaccent{\tilde}{#2}} % if no number is given
    {\ifthenelse{\equal{#1}{*}}
        {\undertilde{#2}} % if * is given
    {\underaccent{\hstretch{#1}{\tilde{}}}{#2}}} % else
}

\newcommand{\range}[2][]{%
    {\mathrm{range} \mspace{-2mu}}
    \ifthenelse{\isempty{#1}}
        {\left( \mspace{1mu} #2 \mspace{1mu} \right)} % if no number is given
    {\ifthenelse{\equal{#1}{*}}
        {\left( \mspace{1mu} #2 \mspace{1mu} \right)} % if * is given
    {\ifthenelse{\equal{#1}{0}}
        {( \mspace{1mu} #2 \mspace{1mu} )} % if 0 is given
    {\ifthenelse{\equal{#1}{1}}
        {\big( \mspace{1mu} #2 \mspace{1mu} \big)} % if 1 is given
    {\ifthenelse{\equal{#1}{2}}
        {\Big( \mspace{1mu} #2 \mspace{1mu} \Big)} % if 2 is given
    {\ifthenelse{\equal{#1}{3}}
        {\bigg( \mspace{1mu} #2 \mspace{1mu} \bigg)} % if 3 is given
    {\ifthenelse{\equal{#1}{4}}
        {\Bigg( \mspace{1mu} #2 \mspace{1mu} \Bigg)} % if 4 is given
    {\PackageError{mycommands}{#1 is an invalid optional argument}{Values should be between 0 and 4}}}}}}}} % else
}

\newcommand{\rank}[2][]{%
    {\mathrm{rank} \mspace{-2mu}}
    \ifthenelse{\isempty{#1}}
        {\left( \mspace{1mu} #2 \mspace{1mu} \right)} % if no number is given
    {\ifthenelse{\equal{#1}{*}}
        {\left( \mspace{1mu} #2 \mspace{1mu} \right)} % if * is given
    {\ifthenelse{\equal{#1}{0}}
        {( \mspace{1mu} #2 \mspace{1mu} )} % if 0 is given
    {\ifthenelse{\equal{#1}{1}}
        {\big( \mspace{1mu} #2 \mspace{1mu} \big)} % if 1 is given
    {\ifthenelse{\equal{#1}{2}}
        {\Big( \mspace{1mu} #2 \mspace{1mu} \Big)} % if 2 is given
    {\ifthenelse{\equal{#1}{3}}
        {\bigg( \mspace{1mu} #2 \mspace{1mu} \bigg)} % if 3 is given
    {\ifthenelse{\equal{#1}{4}}
        {\Bigg( \mspace{1mu} #2 \mspace{1mu} \Bigg)} % if 4 is given
    {\PackageError{mycommands}{#1 is an invalid optional argument}{Values should be between 0 and 4}}}}}}}} % else
}

\renewcommand{\qedsymbol}{{\color{mydarkblue}\scriptsize\ensuremath{\blacksquare}}}

\newcommand\ellipsis{\ifmmode\hspace{0.1em}...\else\makebox[1em][c]{.\hfil.\hfil.}\thinspace\fi}

\newcommand*{\SavedEqref}{}
\let\SavedEqref\eqref
\renewcommand*{\eqref}[1]{%
\begingroup
    \hypersetup{%
        linkcolor=black,
        linkbordercolor=black,
    }%
    \SavedEqref{#1}%
\endgroup
}

\newcommand{\mf}[1]{\textit{\ttfamily #1}}

\renewcommand{\leq}{\leqslant}
\renewcommand{\geq}{\geqslant}
\renewcommand{\succeq}{\succcurlyeq}

\makeatletter
\DeclareRobustCommand{\sqcdot}{\mathbin{\scaleobj{1.2}{\mathpalette\morphic@sqcdot\relax}}}
\newcommand{\morphic@sqcdot}[2]{%
    \sbox\z@{$\m@th#1\centerdot$}%
    \ht\z@=.33333\ht\z@
    \vcenter{\box\z@}%
}
\makeatother

\newcommand{\dtdyn}[1]{#1}

\newsavebox{\foobox}

\newcommand\specialnote[2][3pt]{\smash{%
    \begin{array}{@{}l@{}}
        \\[#1] % empty line
        #2 
    \end{array}
}}

\newcommand{\ul}[1]{\underline{#1}}

\makeatletter
\newcommand{\coloruline}[2]{%
        \UL@protected\def\temp@uline{\leavevmode \bgroup
            \markoverwith{\textcolor{#1}{\rule[-0.7ex]{2pt}{1pt}}}%
        \ULon}%
        \temp@uline{#2}%
    }
    \newcommand{\coloruuline}[2]{%
        \UL@protected\def\temp@uuline{\leavevmode \bgroup
            \UL@setULdepth
            \ifx\UL@on\UL@onin \advance\ULdepth2.8\p@\fi
            \markoverwith{\textcolor{#1}{\lower\ULdepth\hbox
                {\kern-.03em\vbox{\hrule width.2em\kern1\p@\hrule}\kern-.03em}}}%
        \ULon}%
        \temp@uuline{#2}%
    }

    \newcommand{\coloruwave}[2]{%
        \UL@protected\def\temp@uwave{\leavevmode \bgroup 
        \ifdim \ULdepth=\maxdimen \ULdepth 3.5\p@
        \else \advance\ULdepth2\p@ 
        \fi \markoverwith{\textcolor{#1}{\lower\ULdepth\hbox{\sixly \char58}}}\ULon}
        \font\sixly=lasy6 % does not re-load if already loaded, so no memory drain.
        \temp@uwave{#2}%
    }
\makeatother

\makeatletter

\newcommand{\disableAcronymHyperlink}{%
    \def\AC@hyperlink##1##2{##2}%
    \def\AC@hyperref[##1]##2{##2}%
    \def\AC@hypertarget##1##2{##2}%
    \def\AC@phantomsection{}%
}
\makeatother

\newlength{\CapLenGreek}
\AtBeginDocument{\settoheight{\CapLenGreek}{$\mathbb{E}$}}
\newcommand{\ubar}[1]{\bar{#1}}

\makeatletter

\makeatletter

\newcommand{\FIXME}[1]{}

\newcommand{\textcomma}{\text{,}}
\newcommand{\textperiod}{\text{.}}

\newcommand{\sysn}[1]{\mathcal{#1}}
\newcommand{\setn}[1]{\mathbb{#1}}

\newcommand{\dimx}{n}
\newcommand{\dimu}{p}
\newcommand{\dimy}{n}
\newcommand{\nbitsm}{m}
\newcommand{\nbitse}{q}

\newcommand{\nbitstotal}{N}
\newcommand{\bias}{b(\nbitse)}
\newcommand{\smallesstepsize}{\delta}
\newcommand{\windowsize}{r}
\newcommand{\attacklength}{T}
\newcommand{\nbitswm}{L}

\newcommand\x{0}

\let\originalleft\left
\let\originalright\right
\renewcommand{\left}{\mathopen{}\mathclose\bgroup\originalleft}
\renewcommand{\right}{\aftergroup\egroup\originalright}

\newcommand{\recieved}[1]{#1^{\downarrow}}

\crefname{equation}{}{}
\creflabelformat{equation}{\textup{(#2#1#3)}}

\ExplSyntaxOn
\makeatletter
\renewcommand{\@firstupper}[1]{\text_titlecase_first:n {#1}}
\makeatother
\ExplSyntaxOff

\newcommand{\tn}[1]{\textnormal{#1}}

\let\oldsetminus\setminus
\renewcommand{\setminus}{\mathchoice
  {\mkern -1mu\oldsetminus{}\mkern -1mu}%
  {\mkern -1mu\oldsetminus{}\mkern -1mu}%
  {\oldsetminus}%
  {\oldsetminus}%
}

\let\oldfrac\frac
\renewcommand{\frac}[2]{\mathchoice
  {\oldfrac{#1}{#2}}%
  {#1/#2}%
  {#1/#2}%
  {#1/#2}%
}

\DeclareSIUnit\rpm{RPM}

\makeatletter%
\newcommand{\labelgroup}[1]{%
  \cref@label[equation]{#1}{\theequation}%
}
\makeatother%

%% file: Styles/colors.tex
\colorlet{mydarkblue}{blue!30!black}
\colorlet{mydarkred}{red!80!black}
\colorlet{mylightgray}{white!98!black}
\colorlet{mymediumgray}{white!85!black}
\colorlet{mydarkgray}{white!50!black}

\definecolor{tableau1}{HTML}{4E79A7}  % Blue
\definecolor{tableau2}{HTML}{F28E28}  % Orange
\definecolor{tableau3}{HTML}{E15759}  % Red
\definecolor{tableau4}{HTML}{76B7B2}  %
\definecolor{tableau5}{HTML}{59A14F}  % Green
\definecolor{tableau6}{HTML}{EDC948}  % 
\definecolor{tableau7}{HTML}{A0CBE8}  % Light-blue
\definecolor{tableau8}{HTML}{B07AA1}  % Purple

\definecolor{pycharm1}{RGB}{22,43,58}       % Dark gray background
\definecolor{pycharm2}{HTML}{8FA6F5}        % Baby blue
\definecolor{pycharm3}{HTML}{C8957B}        % Teal green
\definecolor{pycharm4}{HTML}{85A658}        % Bright green
\definecolor{pycharm5}{HTML}{4765B0}        % Purple
\definecolor{pycharm6}{HTML}{F2C786}        % Sand yellow

\definecolor{matlab1}{HTML}{F2C786}
\definecolor{matlab2}{HTML}{AA04F9}         % Purple
\definecolor{matlab3}{HTML}{028009}         % Green

\definecolor{arduino1}{HTML}{028009}
\definecolor{arduino2}{HTML}{AA04F9}

\definecolor{definitionbackgroundcolor}{RGB}{250,250,250}
\definecolor{theorembackgroundcolor}{RGB}{250,250,250}

\definecolor{tudelftmain}{HTML}{0076C2}
\definecolor{tudelftaccent}{HTML}{00A6D6}
\definecolor{tudelftlightgray}{HTML}{E7E7E7}
\definecolor{tudelftmediumgray}{HTML}{BABABA}

\pgfplotsset{%
    colormap={inferno}{%
        rgb255=(0, 0, 4)
        rgb255=(12, 8, 38)
        rgb255=(36, 12, 79)
        rgb255=(66, 10, 104)
        rgb255=(93, 18, 110)
        rgb255=(120, 28, 109)
        rgb255=(147, 38, 103)
        rgb255=(174, 48, 92)
        rgb255=(199, 62, 76)
        rgb255=(221, 81, 58)
        rgb255=(237, 105, 37)
        rgb255=(248, 133, 15)
        rgb255=(252, 165, 10)
        rgb255=(250, 198, 45)
        rgb255=(242, 230, 97)
        rgb255=(252, 255, 164)
    }
}

%% file: Styles/environments.tex
\creflabelformat{equation}{(#2#1#3)}
\Crefname{equation}{}{}
\crefname{equation}{}{}

\creflabelformat{section}{\S#2#1#3}
\Crefname{section}{}{}
\crefname{section}{}{}

\creflabelformat{figure}{#2#1#3}
\Crefname{figure}{Fig.}{Fig.}
\crefname{figure}{Fig.}{Fig.}

\creflabelformat{table}{#2#1#3}
\Crefname{table}{Table}{Table}
\crefname{table}{Table}{Table}

\pgfplotsset{%
    every axis plot/.append style={line width=0.4pt}
}

\pgfplotsset{%
    legend style={nodes={scale=0.85, transform shape}},
    legend style={/tikz/every even column/.append style={column sep=1mm}},
    legend style={%
        at={($(0,1) + (1mm,-1mm)$)},
        anchor=north west,
        inner sep=2pt
    }
}

\pgfplotsset{
    legend image code/.code={
        \draw[mark repeat=2,mark phase=2]
        plot coordinates {
        (0cm,0cm)
        (0.15cm,0cm)        %% default is (0.3cm,0cm)
        (0.3cm,0cm)         %% default is (0.6cm,0cm)
        };%
    }
}

\newtheorem{exampleNoBackground}{Example}

\newtheorem{remarkNoBackground}{Remark}

\newenvironment{remark}
    {
    \begin{remarkNoBackground}
        }
        {
    \end{remarkNoBackground}
    }

\creflabelformat{remarkNoBackground}{#2#1#3}
\Crefname{remarkNoBackground}{Rmk.}{Rmk.}
\crefname{remarkNoBackground}{Rmk.}{Rmk.}

\mdfdefinestyle{definition}{
    backgroundcolor = mylightgray,
    leftmargin = 5,
    rightmargin = 5,
    innerleftmargin = 10,
    innerrightmargin = 10,
    innertopmargin = 5,
    innerbottommargin = 5,
    outerlinewidth = 1,
    linecolor = mylightgray}

\newtheorem{definitionNoBackground}{Definition}

\newenvironment{definition}
    {
    \begin{definitionNoBackground}
        }
        {
    \end{definitionNoBackground}
    }

\creflabelformat{definitionNoBackground}{#2#1#3}
\Crefname{definitionNoBackground}{Def.}{Def.}
\crefname{definitionNoBackground}{Def.}{Def.}

\renewenvironment{proof}{{\bfseries \itshape \color{mydarkblue} Proof:}}{}

\AtEndEnvironment{proof}{\null\hfill{\color{mydarkblue}\qedsymbol}\par\medskip}

\mdfdefinestyle{theorem}{
    backgroundcolor = mylightgray,
    leftmargin = 5,
    rightmargin = 5,
    innerleftmargin = 10,
    innerrightmargin = 10,
    innertopmargin = 5,
    innerbottommargin = 5,
    outerlinewidth = 1,
    linecolor = mylightgray}

\newtheorem{theoremNoBackground}{Theorem}

\creflabelformat{theoremNoBackground}{#2#1#3}
\Crefname{theoremNoBackground}{Prop.}{Prop.}
\crefname{theoremNoBackground}{Prop.}{Prop.}

\newtheorem{propositionNoBackground}[theoremNoBackground]{Proposition}

\newenvironment{proposition}
    {
    \begin{propositionNoBackground}
    }
    {
    \end{propositionNoBackground}
    }

\creflabelformat{propositionNoBackground}{#2#1#3}
\Crefname{propositionNoBackground}{Prop.}{Prop.}
\crefname{propositionNoBackground}{Prop.}{Prop.}

\newtheorem{corollaryNoBackground}{Corollary}

\newenvironment{corollary}
    {
    \begin{corollaryNoBackground}
    }
    {
    \end{corollaryNoBackground}
    }

\newtheorem{lemmaNoBackground}[theoremNoBackground]{Lemma}

\creflabelformat{lemmaNoBackground}{#2#1#3}
\Crefname{lemmaNoBackground}{Lem.}{Lem.}
\crefname{lemmaNoBackground}{Lem.}{Lem.}

\newtheorem{assumptionNoBackground}{Assumption}

\newenvironment{assumption}
    {
    \begin{assumptionNoBackground}
        }
        {
    \end{assumptionNoBackground}
    }

\AtEndEnvironment{assumption}{\null\hfill\ensuremath{\lozenge}}

\creflabelformat{assumptionNoBackground}{#2#1#3}
\Crefname{assumptionNoBackground}{Asm.}{Asm.}
\crefname{assumptionNoBackground}{Asm.}{Asm.}
\newtheorem*{problemNoBackground}{Problem statement}

\newenvironment{problem}
    {
    \begin{problemNoBackground}
        }
        {
    \end{problemNoBackground}
    }

\creflabelformat{problemNoBackground}{#2#1#3}
\Crefname{problemNoBackground}{Prb.}{Prb.}
\crefname{problemNoBackground}{Prb.}{Prb.}

%% file: Styles/graphs.tex
\pgfplotsset{
    /pgfplots/layers/Bowpark/.define layer set={
        axis background,axis grid,main,axis ticks,axis lines,axis tick labels,
        axis descriptions,axis foreground
    }{/pgfplots/layers/standard},
}

\pgfplotsset{
    /pgfplots/ylabel style={rotate=-90},
}

\pgfplotsset{
    /pgfplots/grid style={mylightgray},
}

\pgfplotsset{
    every axis plot/.append style={line width=0.4pt}
}

\pgfplotsset{
    legend style={nodes={scale=0.7, transform shape}},
    legend style={/tikz/every even column/.append style={column sep=1mm}}
}

\pgfplotsset{
    every axis/.append style={
        label style={font=\footnotesize},
        tick label style={font=\footnotesize}  
    }
}

%% file: Styles/spacing.tex
\newlength{\beforetabularxspacing}
\newlength{\aftertabularxspacing}
\newcommand{\tikzfont}{\footnotesize}  % For font size in Tikz pictures

\newcommand{\channelabelsize}{\tiny}  % For font size of the channel labels

\newlength{\dynamicswidth}  
\newlength{\dynamicsheigth}
\newlength{\horizontalblocksep}
\newlength{\verticalblocksep}
\newlength{\verticalsependpoints}
\newlength{\headersep}
\newlength{\verticalsep}
\newlength{\networkspacingsep}
\newlength{\networklinewidth}
\newlength{\squareblock}
\newlength{\networkarrowheight}
\newlength{\networkarrowwidth}
\newlength{\summationsize}
\newlength{\switchsize}
\newlength{\attackcirclesize}
\newlength{\horizontalsepinputsmall}  % TTC, STC, etc.
\newlength{\phasespacesize}  
\newlength{\phasespacesideplotwidth}
\newlength{\phasespacesideplotwidthtwo}  % for the paper versions
\newlength{\phasespacesideplotheight}
\newlength{\phasespacesidespacing}
\newlength{\phasespacesidespacingtwo}
\newlength{\zdaplotwidth}
\newlength{\zdaplotheight}
\newlength{\zdaplotsep}
\tikzset{
    pointdown/.style args={#1,#2,#3}{
        code={
            \draw (0,0) -- (1,-1) -- (2,0) -- (2,2) -- (0,2) -- (0,0);
        \node[#3] (#1) at (1,1) {#2};
        }
    }
}

\tikzset{
  mydashdot/.style={dash pattern=on 8pt off 2pt on 1pt off 2pt},
}

\tikzset{
  mynetworkdash/.style={dash pattern=on 1pt off 1pt on 1pt off 1pt},
}

%% file: Sections/ecc_2026.tex
\section{Introduction}\label{sec:introduction}

% Intro ICSs and NCSs
\lettrine{O}{ver} the past decades, the use of communications technology in \acfp{ICS} has seen an exponential increase due to cost and implementation benefits. The resulting \acp{NCS} provide many benefits, but also expose systems to new risks; one of these is their vulnerability to cyberattacks, the number of which targeting critical infrastructure has also risen sharply.  

% Legacy systems and requirements
Traditionally, these~\ac{ICS} architectures, many which were designed decades ago, were constructed without security requirements in mind~\cite{fovinoDesignImplementationSecure2009}. \Acp{ICS} are characterize by stringent requirements on continuous operation and strict real-time constraints, making them difficult to modify. Furthermore, these system use (proprietary) industrial communication protocols that provide no security guarantees.
% such as Modbus-over-TCP (Modbus/TCP)
% , and insecure and outdated industrial communication protocols, are still widely used for their efficiency and ease-of-use

% Types of attacks targetting them
In recent years, due to the aforementioned vulnerabilities, we have seen more attacks targeting~\acp{NCS}, the archetypical example being the \Stuxnet{} worm in 2010~\cite{katulicProtectingModbusTCPBased2023}\FIXME{I removed the citation~\cite{langnerStuxnetDissectingCyberwarfare2011} for space}. In response, this has lead to research from the control community, investigating control-theoretical attacks such as replay attacks~\cite{moPhysicalAuthenticationControl2015} and \acp{ZDA}~\cite{teixeiraSecureControlFramework2015}, leading to the field of secure control. 
% More specifically, we consider \acp{NCS} subject to \ac{MITM} attacks, in particular \ac{FDI} and replay attacks. Other attacks, such as \acp{ZDA} and \ac{DoS} attacks, are out of scope.

% Traditional security and CIA triad
% FIXME: Make this a better bridge.
Conventional~\ac{IT} cybersecurity solutions, such as encryption, focus on the goals of \ac{CIA}, with confidentiality being the primary goal. However, in many~\acp{ICS} with strict real-time constraint, availability takes absolutely priority, with the additional need for authentication~\cite{katulicProtectingModbusTCPBased2023}. Even a intermediate loss of availability, as with encryption and related schemes~\cite{ferrariSwitchingMultiplicativeWatermarking2021}, could lead to problems with synchronization and previously installed monitoring systems~\cite{bernieriTAMBUSNovelAuthentication2020}\FIXME{Removed~\cite{smithCryptographyConceptsEffects2018,katulicProtectingModbusTCPBased2023} for space}. Furthermore, as these~\acp{ICS} consist mostly of legacy systems, controllers and communications protocols have already been designed, allowing only minor modifications with minimal computational overhead. Therefore, data authentication is essential to prevent attacks on~\ac{ICS} communications, but should be developed with the availability requirement in mind~\cite{katulicProtectingModbusTCPBased2023}.
% ~\cite{katulicProtectingModbusTCPBased2023}
% , and should always be preserved
% Without integrity and authenticity, protocols such as Modbus/TCP are vulnerable to \ac{MITM} attacks~\cite{katulicProtectingModbusTCPBased2023}.
% Furthermore, the goal of confidentiality is often associated with encryption~\cite{fovinoDesignImplementationSecure2009}: however, field devices like~\acp{RTU} and~\acp{PLC} do not have the computational capacity to perform significant cryptographic operations on network packets~\cite{hayesSecuringModbusTransactions2013}, and confidentiality controls are therefore often not recommended~\cite{hayesSecuringModbusTransactions2013}.

% Proprosed proir solutions
To address the specific cybersecurity needs of~\acp{ICS}, several solutions have been proposed. \Ac{IT} based literature suggests modifying the communication layer, adding message authentication, or changing header information~\cite{katulicProtectingModbusTCPBased2023,hayesSecuringModbusTransactions2013,fovinoDesignImplementationSecure2009}\FIXME{Removed~\cite{taylorEnhancingIntegrityModbus2017} for space}. From a control-theoretical perspective, previous works have exploited model-based detection techniques, which we can crudely devise into additive watermarking~\cite{moPhysicalAuthenticationControl2015}\FIXME{Removed~\cite{liuOnlineApproachPhysical2020,sanchezDetectionReplayAttacks2019} for space} and multiplicative watermarking~\cite{ferrariSwitchingMultiplicativeWatermarking2021}\FIXME{Removed~\cite{galloDesignMultiplicativeWatermarking2021,fangTwowayCodingControl2019} for space}.
\FIXME{Missing the paper about adding noise and removing noise}
% a predominant communication protocol, or

% Problems with these solutions
The aforementioned solutions do suffer from various drawbacks. The~\ac{IT} solutions all require modifications to the network protocol layer, but this clashes with the requirement of continuous operation, as most~\acp{ICS} cannot be taken offline for maintenance~\cite{hayesSecuringModbusTransactions2013}. Furthermore, several schemes increase the size of each packet~\cite{fovinoDesignImplementationSecure2009,katulicProtectingModbusTCPBased2023}, require that every packet is checked before processing~\cite{katulicProtectingModbusTCPBased2023}, or require additional hardware, introducing additional latency or cost\FIXME{, which might be prohibitive for control loops with small sampling times}. From a control-theoretical perspective, additive watermarking~\cite{moPhysicalAuthenticationControl2015} does not provide authentication guarantees per channel, and detectors require detailed plant knowledge. As for multiplicative watermarking, these schemes rely on a separate channel to perform key synchronization, which often is not available in legacy~\acp{ICS}. Furthermore, if desynchronization ever occurs, the transmitted data becomes incomprehensible\FIXME{Removed footnote above robust controller for space,\footnotemark{}}, which can simply be unacceptable.

To mitigate these issues, we propose a coding scheme that does not suffer from the aforementioned drawbacks.
%We aim to design an extremely lightweight, legacy compatible, authentication scheme prioritizing availability above everything else. The scheme does not require any modification to the communications networks nor the controller. Our solution bridges the gap between control-theoretical methods and protocol layer countermeasures.

\textit{Contributions:} We propose an authentication scheme based on modifying the \acp{LSB} of the measurement signal \FIXME{(i.e., the control and measurement signals)} sent over the communications network.\FIXME{Contrary to other schemes proposed in the literature} The scheme is legacy compatible and does not introduce additional traffic overhead. This means that neither the communications protocol nor the controller needs to be redesigned, and the security is provided as an extra layer on already existing infrastructure. Compared to~\cite{ferrariSwitchingMultiplicativeWatermarking2021,katulicProtectingModbusTCPBased2023,hayesSecuringModbusTransactions2013}, we do not require a separate secure channel for synchronization, but instead propose a rudimentary yet effective look-ahead window resynchronization scheme, capable of handling packet dropouts\FIXME{, the effectiveness of which we demonstrate in simulation}. Most importantly, even if desynchronization occurs, availability is maintained, providing an advantage over other schemes~\cite{ferrariSwitchingMultiplicativeWatermarking2021,fovinoDesignImplementationSecure2009}.
Secondly, whilst~\acp{LSB} coding schemes have been considered before as a use for covert channels~\cite{bernieriTAMBUSNovelAuthentication2020}, to the best of the authors knowledge, we are the first to quantify its effect on control performance\FIXME{, by deriving the appropriate methodology and corresponding metrics}. We do this both for fixed-point and floating-point methodologies, as both implementation need to be considered when dealing with legacy systems. 

\textbf{Notation:} Let $\setn{B}_{N} = \{ \mf{b}_{1} \circ \cdots \circ \mf{b}_{N} \forwhich{} \mf{b}_{i} \in \{\texttt{0}, \texttt{1} \} \}$ (where $\circ$ denotes concatenation) denote the set of bit strings of length $N$, and $\setn{B}_{*}$ the set of bit strings of arbitrary length. For a symmetric matrix $\mtrx{P} = \mtrx{P}^{\T} \in \R^{\dimx \times \dimx}$, let $\mtrx{P} \succ 0$ ($\mtrx{P} \succeq 0$) denote that the matrix $\mtrx{P}$ is~\acf{PD} (\acf{PSD}). Given a vector $\vctr{v} \in \R^{\dimx}$, let $v_{i} \in \R$ denote the $i$-th component of $\vctr{v}$. Let $\vnorm{\vctr{w}(t)}_{\ell_{\infty}} = \sup_{t} \vnorm{\vctr{w}(t)}_{\infty}$ denote the $\ell_{\infty}$ norm of the signal $\vctr{w}(t)$. The $L_{1}$ norm of a system $\mtrx{G}(z)$ is the $\ell_{\infty}$-induced norm given by $\vnorm{\mtrx{G}(z)}_{1} = \sup_{\vnorm{\vctr{w}}_{\infty} \leq 1} \frac{\vnorm{ \mtrx{G}(z) \vctr{w} }_{\infty} }{\vnorm{\vctr{w}}_{\infty} }$.

\section{Problem Formulation}\label{sec:problem_formulation}

Consider the discrete-time~\ac{LTI} system given by:
\begin{subequations}\label{eqn:lti_plant}\begin{align}
    \label{eqn:lti_plant_dynamics}
    \specialnote{\dtdyn{\sysn{P}}:\qquad} \vctr{x}(t+1) &= \mtrx{A} \vctr{x}(t) + \mtrx{B} \vctr{u}(t) + \mtrx{B}_{\tn{w}} \vctr{w}(t) \textcomma{}\\
    \label{eqn:lti_plant_output}
    \vctr{y}(t) &= Q(\vctr{x}(t) ) \textcomma{}
\end{align}\end{subequations}
with physical state $\vctr{x}(t) \in \R^{\dimx}$, measurements $\vctr{y}(t) \in \R^{\dimx}$, and quantizer $Q$, which will be discussed in~\cref{sec:quantization}. The process noise $\vctr{w}(t)$ satisfies:

% \FIXME{Note that we are here assuming the worst-case behavior of the system, unknown disturbances, and no sensor noise (meaning any additional modification we make to the sensor induces a loss of information, and extra error) }

\begin{assumption}[Bounded noise]\label{ass:bounded_disturbance}
    The vector $\vctr{w}(t)$ is bounded in magnitude, meaning $\vnorm{\vctr{w}(t)}_{\ell_{\infty}} \leq \ubar{w}$.
\end{assumption}

% an \ac{RTU} gathers telemetry data from the physical process (e.g., pumps, temperature sensors, pressure sensors, and valves)~\cite{hayesSecuringModbusTransactions2013}

The sensors gather telemetry data from the physical process, and the measurement $\vctr{y}(t)$ is sent over a communications channel to the controller, forming a~\ac{NCS}. For brevity, we consider full-state feedback, although extensions to the more general output-feedback case are possible. The controller $\sysn{C}$ is given by:
\begin{equation}\label{eqn:controller}
    \sysn{C}: \qquad\qquad \vctr{u}(t) = -\mtrx{K} \recieved{\vctr{y}}(t) \textcomma{} \quad\qquad\qquad\strut
\end{equation}
where $\mtrx{K} \in \R^{\dimu \times \dimx}$ is a static gain, $\vctr{u}(t) \in \R^{\dimu}$ is the control signal, and $\recieved{\vctr{y}}(t)$ is the received measurement.

% \FIXME{We consider the objective of regulation, i.e., getting the state $\vctr{x}(t)$ close to the origin, and that of disturbance rejection, i.e., mitigating the impact $\vctr{w}(t)$ has on the deviation of the state from the origin. However, note that it is imperative that we assume the controller has already been designed with these specifications in mind.}

Between the sensors and the controller, an adversary $\sysn{A}$ is present, which can modify the transmitted measurements. Following the taxonomy of~\cite{sandbergSecureNetworkedControl2022}, the received measurement $\recieved{\vctr{y}}(t)$ is given by:
\begin{equation}\label{eqn:transmitted_attacked}
    \sysn{A}: \qquad\qquad \recieved{\vctr{y}}(t) = \vctr{y}(t) + \vctr{a}(t) \textcomma{} \qquad\qquad\strut
\end{equation}
where $\vctr{a}(t)$ is an adversarial signal chosen by the adversary. We model the concept of safety as a convex set $\setn{X}_{\tn{safe}} \subset \R^{\dimx}$. Under nominal conditions, the closed-loop control system has been designed such that $\vctr{x}(t) \in \setn{X}_{\tn{safe}}$ holds for all $t$. This leads to the following attack objective:

%\FIXME{Do we want the assumption that the state is in the safe set under nominal performance? We can make this assumption I think?}

\begin{definition}[Successful attack]
    A~\emph{successful attack} of length $\attacklength$ is defined as $\vctr{x}(\attacklength) \notin \setn{X}_{\tn{safe}}$\FIXME{Do we want a safe set? Could be sufficiently large instead? Safe-set is maybe somewhat bothersome...} (disruptive) and $g(t) = 0$ for $t_{\tn{a}} \leq t < \attacklength$ (stealthy), where $t_{\tn{a}}$ denotes the beginning of the attack, and $g(t)$ is a detection signal (see~{\textup{\cref{sec:lsb_coding}}}).\FIXME{Is this look-ahead annoying?}
\end{definition}

\subsection{Quantization}\label{sec:quantization}

%  on the sensors to be transmitted over the network, and for it to be used by the controller

As the measurement $\vctr{y}(t)$ is a digital signal, the physical state $\vctr{x}(t)$ needs to be quantized. The quantizer $Q : \R \to \setn{F} \subset \R$ is given by:
\begin{equation}\label{eqn:sensor_as_quantizer}
    \sysn{S}: \qquad y(t) = Q(x(t)) = \argmin_{y \in \setn{F} \subset \R} \vert{} x(t) - y \vert{} \textcomma{}
\end{equation}
corresponding to quantization with round-off~\cite{franklinDigitalControlDynamic2002}\FIXME{Do we want to explain this is a footnote, with the infinite base-2 representation?}. Whenever $\vctr{x}(t)$ is a vector, we imply that $Q(\vctr{x}(t))$ is computed element-wise. Here, $\setn{F}$ is a finite, symmetric set of quantization levels, meaning $y \in \setn{F}$ implies $ -y \in \setn{F}$.
% $\mathrm{card}(\setn{F}) < \infty$ and $-\setn{F} = \setn{F}$
% , meaning there are a finite number of quantization levels, and that $y \in \setn{F} \implies -y \in \setn{F}$.\FIXME{Can be shorter maybe}
%  (rather then truncation)

% \begin{remark}\label{rmk:qunatization_and_controller_model_precision_assumptions}
%     In this work, for simplicity, we only consider a~\ac{S2C} channel, as controllers are often connected to plants through wired networks\FIXME{Adress comment of Riccardo about wired.}. On the other hand, the measurement signal are collected by sensors, which may be connected through wireless networks~\cite{ohnoMinMaxDesignFeedback2017}. Thus, here we focus on the quantization error of the measurement signal, assuming that there is  no quantization error at the controller. However, the latter, as well as error induced to fixed/floating-point arithmetic, can be included in the analysis as well, see e.g.~\cite{vanwingerdenInfluenceFiniteWord1984}.\FIXME{Also note: roundoff errors are roughly proportional to the amplitude of the represented quantity.~\cite{widrowQuantizationNoiseRoundoff2008}}
% \end{remark}

% \begin{remark}\label{rmk:qunatization_and_controller_model_precision_assumptions}
%     For brevity, we only consider a~\ac{S2C} channel, as the sensors are often not collocated with the controller~\cite{ohnoMinMaxDesignFeedback2017}, but the framework can readily be expanded to deal with input quantization as well as error introduced due to fixed/floating-point arithmetic (see, .e.g.,~\cite{vanwingerdenInfluenceFiniteWord1984}).
% \end{remark}
\FIXME{For space, I've removed this remark for now....}

The quantization set $\setn{F}$ depends on an implementation with either a fixed-point or floating-point number format. Whilst floating-point is nowadays the \emph{de facto} standard in modern computing~\cite{widrowQuantizationNoiseRoundoff2008}, fixed-point arithmetic is still used in embedded systems\FIXME{Removed~\cite{ohnoMinMaxDesignFeedback2017} for space}, and certain edge devises only support fixed-point arithmetic~\cite{nikolicPerformanceAnalysisTwo2025}. Traditionally,~\ac{A/D} converters are based on uniform quantization\FIXME{, i.e., fixed-point~\cite{kontroFloatingpointArithmeticSignal1992}}, and fixed-point representation is therefore typical in real-time control system\FIXME{said computer before}~\cite{franklinDigitalControlDynamic2002}.
% Similar to~\cite{millerQuantizerEffectsSteadystate1989}, here we will consider both.

% and computation and systems simulation are almost always done in floating-point~\cite{widrowQuantizationNoiseRoundoff2008}
% (such as sensor logic)
% , especially in legacy systems, due to its advantages\FIXME{Chips can be smaller and take up less space. They also use less power. They also require less memory (so also smaller size), and they are faster in their computations. Fewer transistors, so they are also cheaper.}, and the are still used for control signals~\cite{widrowQuantizationNoiseRoundoff2008}, as to why it is of interest to us

%This means that, for fixed-point, if we display $(x)_{2}$ (which might mean we need an infinite number of non-zero bits in the mantissa), that is the $\nbitsd + 1$-the bit is $\mf{1}$, then the $\nbitsd$-th bit is rounded to $\mf{1}$ as well.
% For floating point, the procedure is the same, by first selecting the `correct' exponent. Then we have $\setn{F}_{\tn{FL}}(\nbitse, \nbitsm) \subset \R$ for floating-point and $\setn{F}_{\tn{FX}}$ for fixed-point.

As the quantization levels $\setn{F}$ are symmetric around zero, a single sign bit $\mf{s} \in \setn{B}_{1}$ is used to store $\mathrm{sign}(x) = (-1)^{\mf{s}}$. With some abuse of notation, we use $\nbitse$ to denote the number of bits in the integer or exponent part for fixed-point and floating-point, respectively, and $\nbitsm$ the number of bits in the fractional part or mantissa, respectively. Letting $C(n) = 2^{n} - 1$ denote the number of combinations possible with $n$ bits, the set $\setn{F}_{\tn{FX}}$ for fixed-point implementation is given by
%
% As such, the format can be specified by $(\tn{type},1,\nbitse, \nbitsm)$, where the first denotes the sign bit\FIXME{I think, remove this}. Note that we always assume one sign bit.
%
\begin{equation}
    \setn{F}_{\tn{FX}}\! =\! \{\pm(i + d \ccdot 2^{-\nbitsm}) \forwhich{} i \leq C(\nbitse), d \leq C(\nbitsm) \} \textcomma{}
\end{equation}
and for floating-point, this implementation is given by~\cite{nikolicPerformanceAnalysisTwo2025}\FIXME{Footnote about IEEE\footnotemark{}}
\begin{equation}
    \setn{F}_{\tn{FL}}\! =\! \{ \pm (2^{e - b} \ccdot (1 + p \ccdot 2^{-\nbitsm})) \forwhich{} e \leq C(\nbitse), p \leq C(\nbitsm) \} \textcomma{}
\end{equation}
where $\bias = C(\nbitse) - 1$ is called the \emph{bias}. The quantization error is defined as 
\begin{equation}
    q(x(t)) = x(t) - Q(x(t)) \textperiod{}
\end{equation}
Let $\smallesstepsize \in \setn{F}$ denote the smallest step size, which is given by
\begin{equation}
    \delta_{\tn{FX}} = 2^{-\nbitsm}\textcomma{} \qquad \delta_{\tn{FL}} = 2^{-\bias + 1} \ccdot 2^{-\nbitsm} \textperiod{}
\end{equation}

For a floating-point number format, define $\Delta = 2^{\nbitsm} \ccdot \smallesstepsize_{\tn{FL}}$~\cite{widrowQuantizationNoiseRoundoff2008} as the spacing of the cycles (see~\cref{fig:graph_mapping_quantization}). Note that $\nbitse$ and $\nbitsm$ fulfill different roles in the aforementioned number formats. Particularly, $\nbitse$ determines the largest numbers presentable (dynamic range), whilst $\nbitsm$ determines the resolution of the numbers that can be represented.

% \footnotetext{The IEEE number format defines subnormal, or denormalized numbers whenever $e = 0$, providing gradual underflow to zero, which we ignore in this work.}

% \FIXME{Do I even want to include this?}For floating-point numbers, several commonly used standards exist, such as the IEEE 754-2019 standard implementing a single and double precision format, defined by $(\tn{FL},1,8,23)$ and $(\tn{FL},1,11,52)$. For fixed point, no widely-used standards are available.

\begin{figure}
    \centering%
    \input{Graphs/graph_mapping_quantization}
    \caption{Mapping for fixed-point and floating-point}\label{fig:graph_mapping_quantization}
\end{figure}
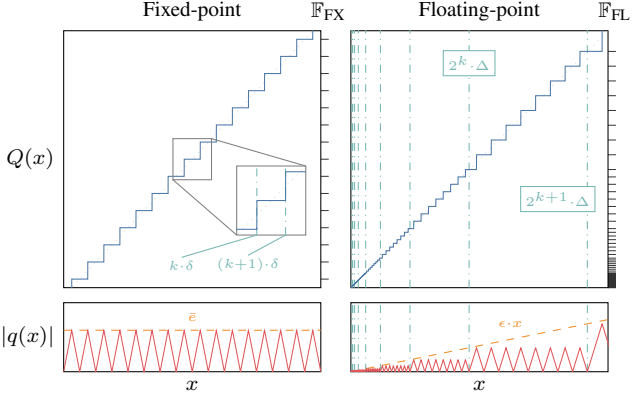

\begin{assumption}[No overflow]\label{ass:no_overflow}
    The number of exponent/integer bits $\nbitse$ is sufficiently large, such that \FIXME{under nominal operation }no overflow occurs. Specifically, for all $\vctr{x}(t) \in \setn{X}_{\tn{safe}}$, $\vnorm{\vctr{x}(t)}_{\infty} \leq 2^{\nbitse} - 2^{-\nbitsm} \approx 2^{\nbitse}$ for fixed-point, and $\vnorm{\vctr{x}(t)}_{\infty} \leq 2^{\bias + 1} \ccdot (2 - 2^{-\nbitsm} ) \approx 2^{\bias + 2}$ for floating point.
\end{assumption}

\Cref{ass:no_overflow} motivates to focus our analysis solely on $\nbitsm$. As $y(t) \in \setn{F}$, it can be represented exactly with $\nbitstotal = 1 + \nbitse + \nbitsm$ bits. We will denote the corresponding bit string of $y(t)$ as $\mf{y}(t) \in \setn{B}_{\nbitstotal}$. In a big-endian system, the rightmost bits of $\mf{y}(t)$ are the~\acp{LSB}. Considering that the adversary $\sysn{A}$ modifies the digital signal in transit (by changing the bits in the representation),\FIXME{ from~\cref{eqn:transmitted_attacked} is must hold that} $\vctr{a}(t)$ is restricted such that $\recieved{\vctr{y}}(t) \in \setn{F}^{\dimy}$. 
% The resulting overall architecture can be seen in~\cref{fig:diagram_ncs_diagram_quantization.tex}.

% (see~\cref{fig:diagram_float_32})

% \begin{figure}[htb]
%     \centering%
%     \input{Diagrams/diagram_float_32}
%     \caption{$y(t) = 13.5390625$ represented in a floating-point format.\FIXME{These number are not true}}\label{fig:diagram_float_32}
% \end{figure}

% \begin{figure}[htb]
%     \centering%
%     \input{Diagrams/diagram_ncs_diagram_quantization.tex}
%     \caption{Considered architecture}\label{fig:diagram_ncs_diagram_quantization.tex}
% \end{figure}

\section{\texorpdfstring{\acsp{LSB} Authentication Scheme}{LSBs Authentication Scheme}}\label{sec:lsb_coding}

To detect attacks on the control system, we propose using the \acp{LSB} of the measurement $\mf{y}(t)$ after $y(t)$ has been quantized. In this section, we consider scalar values, as vector-valued inputs are quantized element-wise\FIXME{(and thus the framework is straightforwardly extended)}. To provide both authentication and integrity, we utilize a~\ac{HMAC}~\cite{hayesSecuringModbusTransactions2013}\FIXME{Removed~\cite{yingCovertChannelBasedTransmitter2022} for space}:\FIXME{Removed~\cite{menezesHandbookAppliedCryptography2018} for space}
% Below, we provide a definition tailored to this work (more detail can be found in~\cite{menezesHandbookAppliedCryptography2018}).

% , we gloss over plenty technical details, for which we refer the reader to[\S{}B.2]\cite{mraihiHOTPHMACBasedOneTime2005}\cite{krawczykCryptographicExtractionKey2010}

% whilst , we use a lightweight (cryptographic) hashing function $H$, specifically an~\ac{HMAC}.\FIXME{For brevity, we gloss over plenty technical details, for which we refer the reader to[\S{}B.2]\cite{mraihiHOTPHMACBasedOneTime2005}\cite{krawczykCryptographicExtractionKey2010}}\FIXME{Note that the definition below are not the defining characteristics of a an HMAC. Instead, they are the properties which are important for our anaylsis.}

% We give a brief, non-technical definition of an~\ac{HMAC} (for further details the reader is referred to~\cite{bellareKeyingHashFunctions1996}, \cite[\S9.5]{menezesHandbookAppliedCryptography2018}).

\begin{definition}[\Ac{HMAC}]\label{def:hmac}
    An~\ac{HMAC} $H : \setn{B}_{*} \times \setn{B}_{K} \to \setn{B}_{D}$ is a deterministic one-way function, mapping a~\emph{query} $(\mf{m},\mf{k})$ to a fixed-length \emph{digest} $\mf{h} = H(\mf{m},\mf{k})$.
    \begin{enumerate}[i)]
        \item{}\emph{Uniformity\FIXME{Removed~\cite[\S9.7.1]{menezesHandbookAppliedCryptography2018} for space}:} The probability of a particular digest $h \in \setn{B}_{D}$ for a given key $\mf{k}$ and a random fixed-length message $\mf{m} \in \setn{B}_{\ell}$ is close to $2^{-D}$.
        \item{}\emph{Key non-recovery\FIXME{Removed~\cite[\S9.2.2]{menezesHandbookAppliedCryptography2018} for space}:} Given the digest $\mf{h} = H(\mf{m}, \mf{k})$ and message $\mf{m}$, it is computationally infeasible to reconstruct the key $\mf{k}$.
        \item{}\emph{Avalanche effect:} Changing a single bit in either $\mf{m}$ or $\mf{k}$ results in about half the bits being changed in the resulting digest $\mf{h}'$ compared to the original digest $\mf{h}$.
    \end{enumerate}
\end{definition}

At the sensors $\sysn{S}$, a digest of the current measurement $\mf{y}_{i}(t)$ is computed using an incremental key $\mf{k}_{\tn{s}, i}(t)$, resulting in:
%
% The authentication scheme we propose is as follows. After the sensor have quantized the state $x(t)$ using either a fixed-point of floating-point implementation, resulting in a number $\mf{y}(t)$, we take the first $1 + \nbitse + \nbitsm - L$ of $\mf{y}(t)$, resulting in $\mf{z}(t)$. Then, we compute
% %
% \begin{equation}
%     \mf{c}(t) = H(\mf{z}(t), \mf{key}_{\tn{s}}(t) ) \textcomma{} 
% \end{equation}
% %
% where $\mf{key}_{\tn{s}}(t)$ is an incremental key which is shared between the detector $\sysn{D}$ and sensors $\breve{\sysn{S}}$. We replace the last $L \leq \nbitsm$ bits of $\mf{y}(t)$ with the first $L$ digits of $\mf{h}_{i}(t)$, resulting in $\breve{\mf{y}}(t)$, which is then sent over the network (see~\cref{fig:diagram_networked_control_with_watermarking}).\FIXME{Note that the update of the \ac{PLC} can be done online, without interrupting operations.} The resulting sensor $\breve{\sysn{S}}$ is given by:
%
\begin{subequations}\label{eqn:modified_sensors}\begin{align}
    \mf{y}_{i}(t) &= Q(x_{i}(t))\FIXME{Notation base two is not totally correct} \textcomma{} \\
    \label{eqn:modified_sensors_digest}
    \breve{\sysn{S}}: \qquad\qquad\mf{h}_{i}(t) &= H([\![\mf{y}_{i}(t)]\!]_{1:(N-L)}, \mf{k}_{\tn{s},i}(t) ) \textcomma{} \quad\strut \\
    \label{eqn:modified_sensors_last_conc_step}
    {\breve{\mf{y}\hphantom{.}}\!}_{\!i}(t) &= [\![\mf{y}_{i}(t)]\!]_{1:(N - L)} \circ [\![\mf{h}_{i}(t)]\!]_{1:L} \textperiod{}  
\end{align}\end{subequations}
Here, ${[\![ \mf{x} ]\!]}_{a:b}$ denotes the bit string between (and including) $a$ and $b$. The resulting modified measurement $\breve{\vctr{y}}(t)$ in~\cref{eqn:modified_sensors_last_conc_step} is then sent over the network.

\begin{remark}
    Note that time-varying keys and a different key per component $i$ is necessary to adequately defend against replay attacks\FIXME{Removed for space~\cite{moPhysicalAuthenticationControl2015}} and routing attacks\FIXME{Removed for space~\cite[\S{}V.A]{milosevicEstimatingImpactCyberAttack2020}}, respectively.
%     %
%     \begin{itemize}
%         \item{}\textbf{\Acf{FDI} attack~\cite{teixeiraSecureControlFramework2015}:} If a regular keyless digest $H(\mf{y}_{i}(t))$ is used, as by~\cref{ass:key_secrecy} $\breve{\sysn{S}} \in \setn{I}_\tn{a}(t) \implies H \in \setn{I}_{\tn{a}}(t)$, the adversary $\sysn{A}$ can simply compute $H(\mf{y}_{i}(t) + \mf{a}_{i}(t))$ and replace the last $L$ bits of the message with the digest, which will not raise an alarm at $\sysn{D}$. 
%         \item{}\textbf{Replay attack:} Considering that an key-based~\ac{HMAC} function is used, but with a static key, as according to~\cref{ass:key_secrecy}, we have $\vctr{y}(t') \in \setn{I}_{\tn{a}}(t)$ for all $t' < t$. As such, replacing $\vctr{y}(t)$ with $\vctr{y}(t')$ will not raise an alarm at~$\sysn{D}$.
%         \item{}\textbf{Routing attack:} A routing attack exchanges component $i$ and $j$ in a vector-valued signals, such that $\mf{y}_{i}(t) = \mf{y}_{j}(t)$ and vise versa~\cite[\S{}V.A]{milosevicEstimatingImpactCyberAttack2020}. If the keys are time-varying but shared for all components $i$, then due to the birthday paradox, after $t \approx \frac{2^{L +1}}{\dimx \ccdot (\dimx + 1)}$\FIXME{This is from Gemini, not sure if its correct} time steps on average, the vector $\vctr{y}(t)$ will have two entries with an identical hash, which can be swapped without $\sysn{D}$ raising an alarm.
%     \end{itemize}
% \end{remark}
\end{remark}

As the fixed-length measurement $\mf{y}_{i}(t) \in \setn{B}_{N-L}$ is fed into the~\ac{HMAC},~\cref{def:hmac}.i implies the output of the~\ac{HMAC} closely resembles a uniform distribution. For analysis purposes, it is therefore customary\FIXME{Do we need a source here?} to make the following assumption:

% \FIXME{REALLY IMPORTANT: A perfect hash gives a uniform output over a static, fixed set of inputs: this set of inputs is the $8$ or $16$ or $32$ bit fixed/floating-point number}

\begin{assumption}[Random oracle~\protect{\cite{mraihiHOTPHMACBasedOneTime2005}}]\label{ass:random_oracle}
    We model $H$ in~\textup{\cref{eqn:modified_sensors_digest}} as a \emph{random oracle}, a deterministic function mapping a query $(\mf{m}, \mf{k})$ to a digest $\mf{h}$ chosen uniformly from its output domain $\setn{B}_{D}$. Repeated queries are mapped to the same digest.
\end{assumption}

Note that~\cref{ass:random_oracle} also justifies the truncation of the digest $\mf{h}_{i}(t)$ to the first $L$ bits as in~\cref{eqn:modified_sensors_last_conc_step}. Evidently, the information in the last $L$~\acp{LSB} is lost, and replaced by random noise, as far as the control system is concerned. At the controller $\sysn{C}$, the received signal $\recieved{\breve{\vctr{y}}}(t)$ is processed as normal. The received measurement is also fed to a detector $\sysn{D}$, defined as follows:
%
% \FIXME{Note that this also means every bit in the output of $H$ has a equal chance of being a $\texttt{0}$ or a $\texttt{1}$. This motivates our assumption that we can truncate the digest.}

% \begin{figure}[htb]
%     \centering%
%     \input{Diagrams/diagram_number_format_with_watermark}
%     \caption{Representation of $\breve{\mf{y}\hphantom{.}}\!(t)$}\label{fig:diagram_number_format_with_watermark}
% \end{figure}
%
\begin{equation}
    \sysn{D}: \quad g(t) = \begin{cases}
        0 \textcomma{} & H(\recieved{\mf{m}}_{i}(t), \mf{k}_{\text{d},i}(t)) = \recieved{\mf{d}}_{i}(t) \textcomma{} \\
        1 \textcomma{} & \text{otherwise} \textcomma{}
    \end{cases}
\end{equation}
where $\recieved{\mf{m}}_{i}(t) = {[\![\recieved{{\breve{\mf{y}\hphantom{.}}\!}}_{\!i}(t)]\!]}_{1:(N-L)}$ denotes the received message, $\recieved{\mf{d}}_{i}(t) = {[\![\recieved{{\breve{\mf{y}\hphantom{.}}\!}}_{\!i}(t)]\!]}_{(N-L+1):N}$ denotes the received digest, and $g(t) = 1$ raises an alarm. Similar to~\cite{bernieriTAMBUSNovelAuthentication2020,katulicProtectingModbusTCPBased2023}\FIXME{Removed~\cite{yingTACANTransmitterAuthentication2019} for space}, we suppose that $\mf{k}_{\tn{s},i}(0) = \mf{k}_{\tn{d},i}(0)$ are pre-shared for all $i$, and that the keys are updated after every message sent. This can be achieved by a~\ac{KDF} $K$, where\FIXME{Made this shorter~\cite[Definition 5]{krawczykCryptographicExtractionKey2010}}\FIXME{Maybe remove this for space?} $\mf{k}(t+1) = K(\mf{k}(0) , \ell) $~\cite{krawczykCryptographicExtractionKey2010}. Here, $\ell$ is an iteration counter, where ideally $\ell = t$ when no packet dropouts have occurred. The full scheme, with modifications highlighted in blue, can be seen in~\cref{fig:diagram_networked_control_with_watermarking}.

\begin{figure}[htb]
    \centering%
    \input{Diagrams/diagram_networked_control}
    \caption{Overview of the~\ac{NCS} with~\acp{LSB} authentication scheme\FIXME{Replace this with $\sysn{S}'$ instead?}}\label{fig:diagram_networked_control_with_watermarking}
\end{figure}
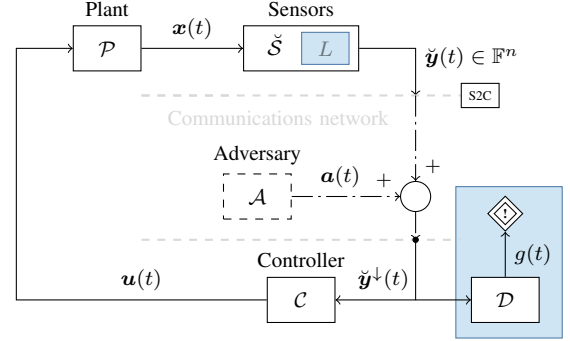

The $L$-bit coding scheme can be implemented on existing hardware, as \ac{HMAC} compatible hashing schemes for\FIXME{removed this'~\ac{FPGA}~\cite{deepakumaraFPGAImplementationMD52001} and microcontrollers' for space} embedded system are available~\cite{mouhaChaskeyEfficientMAC2014}. Furthermore, the proposed scheme does not require a controller redesign.
% , and only slight modifications to sensors, the software of which can be updated without interrupting continuous operation. 
% The proposed transition from $\sysn{S}$ to $\breve{\sysn{S}}$ can be done without interrupting continuous operation, as it constitutes a change in the programmable logic.
% The former is the case as there is no modification to the communications protocol (rather, its only the data which is directly modified), meaning its also protocol agnostic, which is good as these are often proprietary~\FIXME{SOURCE}.\FIXME{Also way too verbose according to Riccardo, we may gloss over these details}
% \FIXME{Also check out the 32 bit MAA as mentioned in~\cite{menezesHandbookAppliedCryptography2018}}

The proposed~\ac{HMAC} scheme can be used to detect~\ac{MITM} attacks even if the adversary is aware of the $L$-bit coding scheme. As such, the protocol is \emph{secure-by-design}, rather than providing security through obscurity by means of a covert channel. 
% approaches do (which can be easily bypassed~\cite{bernieriTAMBUSNovelAuthentication2020}).

%\FIXME{Here we also need to talk about the quantifiable relationship with the number of bits $L$ and how easy it is to forge attacks}\FIXME{Also mention: Kerckhoffs's principle, also metioned in~\cite{sandbergSecureNetworkedControl2022}} This leads to the following assumption:

\begin{assumption}[Kerckhoffs's principle~\cite{sandbergSecureNetworkedControl2022}]\label{ass:key_secrecy}
    The information $\setn{I}_{\tn{a}}(t)$ known to the adversary $\sysn{A}$ satisfies $\{\vctr{y}(0),\ldots, \vctr{y}(t) \land \breve{\sysn{S}}, \sysn{D} \} \subseteq \setn{I}_{\tn{a}}(t)$, $\mf{k}_{\tn{s},i}(t) \notin \setn{I}_{\tn{a}}(t)$, for all $i$, meaning the adversary $\sysn{A}$ can eavesdrop on the measurements $\vctr{y}(t)$ and is aware of the $L$-bit coding scheme, but the keys $\mf{k}_{\tn{s},i}(t)$ are assumed secret.
\end{assumption}

\subsection{Synchronization}

One challenge is ensuring synchronization between the keys $\mf{k}_{\tn{s},i}(t)$ and $\mf{k}_{\tn{d},i}(t)$. Whilst ideally both would update simultaneously, due to real-life network induced phenomena such as packet dropouts, desynchronization might occur.\FIXME{As such, the key at the sensors might update, whilst the key at the detector lags behind.}\FIXME{I removed all of this... 'This problem is not addressed in similar works~\cite{bernieriTAMBUSNovelAuthentication2020}, where synchronization is simply assumed, and~\cite{ferrariSwitchingMultiplicativeWatermarking2021}, where packet dropout is not considered. This provides another motivation why we implement our scheme at the data level, rather than the protocol level.'}
% (even though we modify data, we ensure availability).
% \FIXME{This is a bit weird, because according to~\cite{katulicProtectingModbusTCPBased2023}, most use ModBus TCP, and according to Gemini, these have a three-way handshake.}

One simple yet effective solution is to implement a look-ahead window of size $\windowsize$~\cite{mraihiHOTPHMACBasedOneTime2005}\FIXME{Removed chapter~\cite[\S7.4]{mraihiHOTPHMACBasedOneTime2005}}\FIXME{as also mentioned in~\cite{katulicProtectingModbusTCPBased2023}}. Whenever $H(\recieved{\mf{m}}_{i}(t), \mf{k}_{\text{d},i}(t)) \neq \recieved{\mf{d}}_{i}(t)$, instead of immediately raising an alarm, the values $H(\recieved{\mf{m}}_{i}(t), \mf{k}_{\text{d},i}(t + \tau)) = \recieved{\mf{d}}_{i}(t)$, with $\tau \leq \windowsize$, are also checked, considering a packet dropout might have occurred. Most importantly, even if $\mf{k}_{\text{d},i}(t) \neq \mf{k}_{\text{s},i}(t)$ due to desynchronization, availability is not lost and the control system will continue to operate nominally. The detector $\sysn{D}$ will raise a false alarm, which is undesirable, but it not detrimental for a control system with continuous operation requirements.
% The detector $\sysn{D}$ will raise a false alarm, which need to be inspected, possibly by human staff, but most importantly, it is not detrimental for the control system with continuous operation requirements.

Evidently, the former look-ahead scheme increases the likelihood of a successful attack, whilst a larger $L$ is expected to decrease its likelihood.

% Specifically, under~\cref{ass:random_oracle}, the chance of a successful attack of length $k$ is given by $1 - (1 - 2^{-L})^{\windowsize}$.

% Due to the key-based design of an~\ac{HMAC} and~\cref{ass:key_secrecy}, the scheme is practically resistant against brute force attacks. However, the design still allows for attack with success rates that are probabilistic in nature.

%\FIXME{This kind of proof is related to~\cite[Definition 14]{krawczykCryptographicExtractionKey2010}}

\begin{proposition}\label{thm:security_hmac}
    Suppose~\textup{\cref{ass:random_oracle}} and~\textup{\cref{ass:key_secrecy}} hold. Then, the best strategy of the adversary $\sysn{A}$ guaranties a probability of a successful attack of length $T$ of at most $(1 - (1 - 2^{-\nbitswm })^{\windowsize})^{\attacklength}$.
\end{proposition}

% \begin{proof}
%     See proof of~\cite[Proposition 1]{wolleswinkelQuantifyingControlPerformance2026a}.
% \end{proof}

\input{Sections/arXiv/proof_prop_1}

Having quantified security, we turn our attention to the impact on control performance. We define the virtual performance output
\begin{equation}\label{eqn:virtual_performance_output}
    \vctr{z}(t) = \mtrx{Q}^{\frac{1}{2}} \vctr{x}(t) \textcomma{}
\end{equation}
where $\mtrx{Q} \succ 0$ is a matrix indicating relative importance of each state. Due to the $L$-bit coding scheme, the error under the modified sensor~$\breve{\sysn{S}}$ as in~\eqref{eqn:modified_sensors} is given by

% Note that success probability exponentially decreases as a function of the number of bits $L$\FIXME{What about multiple states???}, indicating why more bits are important. Evidently, if $L = 0$, the attack will guarantied to be successful\FIXME{Should we mention there are no other countermeasures in place?}. Using more then $L=1$ bits is motivated in~\cite{bernieriTAMBUSNovelAuthentication2020}, since in several circumstances the sensibility of sensors do not go beyond $10^{-3}$ (meaning there is actually no loss of information). 
%
\begin{equation}\label{eqn:error_two_independent_source}
    \vctr{e}(t) = \vctr{x}(t) - \breve{Q}(\vctr{x}(t)) = \vctr{x}(t) - (Q(\vctr{x}(t)) + \vctr{d}(t)) \textcomma{}
\end{equation}
where $\vctr{d}(t)$ is the additional error introduced by the coding scheme, and the modified quantizer is denoted by $\breve{Q}$. This leads to the following problem statement:

\begin{problem}
    Consider the closed-loop system as in~\textup{\cref{fig:diagram_networked_control_with_watermarking}}, with the original sensors $\sysn{S}$ as in~\textup{\cref{eqn:sensor_as_quantizer}} updated to $\breve{\sysn{S}}$ as in~\eqref{eqn:modified_sensors}. How do we quantify the impact of the $L$-bit coding scheme on the virtual performance output $\vctr{z}(t)$ as in~\textup{\cref{eqn:virtual_performance_output}}?
\end{problem}

% \FIXME{Include: The above allows the system designed to make a trade-off between $L$ and $r$. Specifically, if it is known that the smallest $k$ needed for a successful attack s $\lbar{k}$ (or an estimate/lower bound is known), and the accepted probability of an attack is given a $1 \ccdot 10^{-6}$, we can decide how many $L$ we need. If we also have some estimate on the probability $\mathbb{P}[\sum_{i} \Xi(t) > r]$, data we gather from the network, we can estimate our false alarm probability. If we also have some (monetary) cost on performance loss, and some cost on false alarm detection, we can do that.}

\section{Control Performance}\label{sec:control_performance}

In order to quantify loss in control performance, we first need to identify an appropriate metric $\rho$. A suitable metric will depend on our implementation of either fixed-point or floating-point, due to qualitative differences in closed-loop behavior. Due to space constraints, we omit the effect of packet dropouts from this analysis, which we leave to future work.
%which will turn out to depend whether our application uses a fixed-point of floating-point format. Note that whilst we may be able to ignore the effect of quantization errors with the implementation by floating-point arithmetic with a sufficient number of bits, there still remain requirements on the implementation by fixed-point arithmetic with a small number of bits due to, e.g., the hardware limitations and the implementation and running costs~\cite{ohnoStateSpaceRealizations2019}. Whilst for scientific computing these errors are usually negligible, They could be significant using a computer with 8 bits or less~\cite{franklinDigitalControlDynamic2002}.\FIXME{This is kind of a modified quote... Riccardo: which controller uses 8 bits?}

\FIXME{Maybe here introduce the noise model? Motivation, they do it in~\cite{bessaFormalNonFragileStability2017} as well. Also, we're interested in worst-case, as that's the strongest requirement.}

\subsection{Fixed-point}

When implementing a fixed-point quantizer $Q_{\tn{FX}}$, it is well know that the resulting (nonlinear) closed-loop system will exhibit undesirable limit cycles\FIXME{Removed \FIXME{Removed~\cite{fuSectorBoundApproach2005} for space}, which is undesirable}. As such, a useful indication of control performance is the size of the smallest set that bounds these limit cycles. Therefore, we turn to invariant ellipsoidal sets given by
\begin{equation}
    \setn{X}_{\tn{reach}} = \{ \vctr{z}(t) \in \R^{\dimx} \forwhich{} \vctr{z}(t)^{\T} \mtrx{E} \vctr{z}(t) \leq 1 \}\textcomma{}
\end{equation}
defined by some matrix $\mtrx{E} \succ 0$, where $\vctr{z}(t) \in \setn{X}_{\tn{reach}} $ implies $ \vctr{z}(t + \tau) \in \setn{X}_{\tn{reach}}$ for all $\tau \geq 0$. Our metric will be the size of this invariant ellipsoid, which is given by
\begin{equation}\label{eqn:set_ellipsoid}
    \rho_{\tn{FX}}(\nbitsm, \nbitswm) = \textstyle\frac{V(\dimx)}{\sqrt{\mathrm{det}(\mtrx{E}^{\star})}} \textcomma{} \FIXME{ Should we use \rho(L \forwhich{} \nbitsm ) instead? }
\end{equation}
where $V(\dimx)$ is the volume of the $n$-ball\FIXME{Removed~\cite[\S8.3.1]{boydConvexOptimization2004} for space}, and $\mtrx{E}^{\star} \succ 0$ is the matrix corresponding to the ellipsoid of minimal volume.

\begin{proposition}\label{thm:volume_ellipsoid}
    Consider a fixed-point quantizer $\breve{Q}_{\tn{FX}}$ with $\nbitsm$ decimal bits and the $L$-bit coding scheme. Then,
    \begin{equation}
        \rho_{\tn{FX}}(\nbitsm, \nbitswm) = \textstyle\frac{V(\dimx)}{\sqrt{\mathrm{det}(\mtrx{P}) \ccdot \mathrm{det}{(\mtrx{Q})}}} \textcomma{}
    \end{equation}
    where the minimum volume ellipsoid that contains the limit cycles caused by $\breve{Q}_{\tn{FX}}$ is given by the matrix $\mtrx{P} \succ 0$, which is a solution to the~\ac{BMI}
    \begin{subequations}\label{eqn:bmi_ellipsoid}\begin{align}
        &\qquad\min_{\alpha,\mtrx{P}} -\log \mathrm{det} \mtrx{P} \quad \tn{s.t.} \quad \alpha \in (0,1) \textcomma{} \\
        \label{eqn:bmi_ellipsoid_bmi}
        &\begin{bmatrix}
            \alpha \ccdot \mtrx{P} - \mtrx{A}_{\tn{cl}}^{\T} \mtrx{P} \mtrx{A}_{\tn{cl}} & -\mtrx{A}_{\tn{cl}}^{\T} \mtrx{P} \bar{\mtrx{B}} \\
            -\bar{\mtrx{B}}^{\T} \mtrx{P} \mtrx{A}_{\tn{cl}} & (1 - \alpha) \ccdot \mtrx{R} - \bar{\mtrx{B}}^{\T} \mtrx{P} \bar{\mtrx{B}}
        \end{bmatrix} \succeq 0 \textcomma{}
    \end{align}\end{subequations}
    where $\mtrx{A}_{\tn{cl}} = \mtrx{A} - \mtrx{B} \mtrx{K}$, $\bar{\mtrx{B}} = \begin{bmatrix} -\mtrx{B} \mtrx{K} & \mtrx{B}_{\tn{w}} \end{bmatrix}$, and $\mtrx{R} = \mathrm{diag}(\frac{1}{\ubar{e}^{2}}, \ldots, \frac{1}{\ubar{w}}^{2}, \ldots)$, with error bound
    \begin{equation}\label{eqn:upper_bound_quantization_error}
        \ubar{e} = 2^{-(\nbitsm + 1)} + 2^{-(\nbitsm - L)} - 2^{-\nbitsm} \approx 2^{-(\nbitsm - L)} \textcomma{}
    \end{equation}
    implying $\rho_{\tn{FX}}(\nbitsm, L) \approx \rho_{\tn{FX}}(\nbitsm - L)$.
\end{proposition}

% \begin{proof}
%     See proof of~\cite[Proposition 2]{wolleswinkelQuantifyingControlPerformance2026a}.
% \end{proof}

\input{Sections/arXiv/proof_prop_2}
\input{Sections/arXiv/corollary}

\subsection{Floating-point}

For floating point, the dynamic range induced by the quantizer $Q_{\tn{FL}}$ can amplify the process noise $\vctr{w}(t)$, and thereby its impact on the virtual performance output $\vctr{z}(t)$\FIXME{Here, we need stronger motivation of dynamic range.}. Defining $\mtrx{G}(z)$ as the closed-loop transfer function from $\vctr{w}(t)$ to $\vctr{z}(t)$, the metric of interest is
%
% and as~\cref{ass:bounded_disturbance} implies $\vnorm{\vctr{z}(t)}_{\ell_{\infty}} \leq \vnorm{\mtrx{G}(z)}_{1} \ccdot \vnorm{\vctr{w}(t)}_{\ell_{\infty}}$
%
% \begin{equation}
%     \rho_{\tn{FL}}(\nbitsm, \nbitswm) = \vnorm{\mtrx{G}(z)}_{1} = \sup_{\vnorm{\vctr{w}}_{\infty} \leq 1} \frac{\vnorm{ \mtrx{G}(z) \vctr{w} }_{\infty} }{\vnorm{\vctr{w}}_{\infty} } \textcomma{}
% \end{equation}
\begin{equation}\label{eqn:metric_floating_gain}
    \rho_{\tn{FL}}(\nbitsm, \nbitswm) = \vnorm{\mtrx{G}(z)}_{1} \textcomma{}
\end{equation}
the worst-case amplification of the process noise to the performance output.

\begin{proposition}\label{thm:induced_gain}
    Consider a floating-point quantizer $\breve{Q}_{\tn{FL}}$ with $\nbitsm$ mantissa bits and the $L$-bit coding scheme. Then,
    \begin{equation}
        \rho_{\tn{FL}}(\nbitsm, \nbitswm) = \frac{\vnorm{\mtrx{Q}}_{\infty} \ccdot \vnorm{(z \ccdot \mtrx{I} - \mtrx{A}_{\tn{cl}})^{\shortminus{}1} \mtrx{B}_{\tn{w}}}_{1}}{1 - \gamma_{\tn{e}} \ccdot \vnorm{(z \ccdot \mtrx{I} - \mtrx{A}_{\tn{cl}})^{\shortminus{}1} \mtrx{B} \mtrx{K}}_{1}} \textcomma{}
    \end{equation}
    where 
    \begin{equation}
        \gamma_{\tn{e}} = 2^{-(\nbitsm + 1)} + 2^{-(\nbitsm - L)} - 2^{-\nbitsm} \approx 2^{-(\nbitsm - L)} \textcomma{}
    \end{equation}
    implying $\rho_{\tn{FL}}(\nbitsm, L) \approx \rho_{\tn{FL}}(\nbitsm - L)$.
\end{proposition}

% \begin{proof}
%     See proof of~\cite[Proposition 3]{wolleswinkelQuantifyingControlPerformance2026a}.
% \end{proof}

\input{Sections/arXiv/proof_prop_3}

The prior results allow us to quantify the worst-case performance loss due to the $L$-bit coding scheme. Similarly, we can quantify the impact in the average sense, by modeling quantization error as sources of random noise\FIXME{Removed~\cite{widrowStatisticalTheoryQuantization1996} for space}\FIXME{Removed~\cite{sripadQuantizationErrorsFloatingpoint1978} for space}, the so-called~\ac{PQN} model~\cite{widrowQuantizationNoiseRoundoff2008}. We define the metric
\begin{equation}\label{eqn:distribution_cost}
    J(\nbitsm, \nbitswm) = \mathbb{E}\Bigg[\sum_{t=0}^{\infty} \vctr{z}(t)^{\T} \vctr{z}(t)\Bigg] \textcomma{}
\end{equation}
leading to the following assumption on the process noise: 

%\FIXME{About uniform, note that whilst it's worst-case, they state:``This bound is only likely to be approached when a system has a constant input and has settled to its steady-state value.''~\cite{franklinDigitalControlDynamic2002} We are in steady-state!}

\begin{assumption}\label{ass:stochastic_process_noise}
    The~\acs{i.i.d.} process noise satisfies $\mathbb{E}[\vctr{w}(t)] = \vctr{0}$, $\mathbb{E}[\vctr{w}(t) \vctr{w}(t)^{\T}] = \mtrx{\Sigma}_{\tn{w}}$, and $\mathbb{E}[\vctr{e}(t) \vctr{w}(t)^{\T}] = \vctr{0}$.
\end{assumption}

Combining the~\ac{PQN} model with~\cref{ass:stochastic_process_noise}, the quantization error $\vctr{q}(t) = \vctr{x}(t) - Q(\vctr{x}(t))$ is modeled to good accuracy as coming from a distribution with zero mean\FIXME{Removed section~\cite[\S12.7]{widrowQuantizationNoiseRoundoff2008}} and covariance $\mtrx{\Sigma}_{\tn{q}} = \mathbb{E}[ \vctr{q}(t)\vctr{q}(t)^{\T} ]$, given by  
\begin{equation}\label{eqn:quantization_error_statistics}
    \mtrx{\Sigma}_{\tn{q}} = \frac{1}{12} \ccdot 2^{-2 \ccdot (\nbitsm + 1)} \ccdot \mtrx{I} \textcomma{} \qquad%
    \mtrx{\Sigma}_{\tn{q}} \approx 0.180 \ccdot 2^{-2\ccdot \nbitsm} \ccdot \mtrx{\Sigma}_{\tn{x}} \textcomma{}
\end{equation}
for a fixed-point\FIXME{~\cite[\S4.2]{widrowQuantizationNoiseRoundoff2008}} and floating-point\FIXME{~\cite[Eq. 12.24]{widrowQuantizationNoiseRoundoff2008}} number format, respectively~\cite{widrowQuantizationNoiseRoundoff2008}. Here, $\mtrx{\Sigma}_{\tn{x}} = \mathbb{E}[\vctr{x}(t)\vctr{x}(t)^{\T}] $ denotes the stationary covariance matrix of $\vctr{x}(t)$.

\begin{proposition}
    Consider the fixed-point or floating-point quantizer $\breve{Q}$ and the $L$-bit coding scheme. Then,
    \begin{equation}
        J(\nbitsm, L) = \mathrm{trace}(\mtrx{Q}\mtrx{\Sigma}_{\tn{x}}) \textcomma{}
    \end{equation}
    where the covariance matrix $\mtrx{\Sigma}_{\tn{x}} \succ 0$ is the solution to the generalized Lyapunov equation
    \begin{equation}\label{eqn:generalized_lyaponuv}
        \mtrx{A}_{\tn{cl}} \mtrx{\Sigma}_{\tn{x}} \mtrx{A}_{\tn{cl}}^{\T} + \mtrx{B} \mtrx{K} \mtrx{\Sigma}_{\tn{e}} (\mtrx{B} \mtrx{K})^{\T} + \mtrx{B}_{\tn{w}} \mtrx{\Sigma}_{\tn{w}} \mtrx{B}_{\tn{w}}^{\T} - \mtrx{\Sigma}_{\tn{x}} = \mtrx{0} \textcomma{}
    \end{equation}
    %
    % with $\mtrx{\Sigma}_{\tn{e}} \approx \frac{1}{12} \ccdot \ubar{e}^{2} \ccdot \mtrx{I}$ for a fixed-point number format, and $\mtrx{\Sigma}_{\tn{e}} \approx (0.180 \ccdot 2^{-2 \ccdot \nbitsm } + \frac{1}{12} \ccdot 2^{-2\ccdot(\nbitsm - L)}) \ccdot \mtrx{\Sigma}_{\tn{x}}$ for a floating-point number format.
    with
    \begin{equation*}
        \mtrx{\Sigma}_{\tn{e}} \approx \frac{1}{12} \ccdot \ubar{e}^{2} \ccdot \mtrx{I} \textcomma{} \;\;\, \mtrx{\Sigma}_{\tn{e}} \approx (0.180 \ccdot 2^{-2 \ccdot \nbitsm } + \frac{1}{12} \ccdot 2^{-2\ccdot(\nbitsm - L)}) \ccdot \mtrx{\Sigma}_{\tn{x}}
    \end{equation*}
    for fixed-point and floating-point, respectively.
    
    % \begin{equation}
    %     \mtrx{\Sigma}_{\tn{e}} \approx
    %     \begin{cases}
    %          \frac{1}{12} \ccdot \ubar{e}^{2} \ccdot \mtrx{I} \textcomma{} & ... \\
    %          (0.180 \ccdot 2^{-2 \ccdot \nbitsm } + \frac{1}{12} \ccdot 2^{-2\ccdot(\nbitsm - L)}) \ccdot \mtrx{\Sigma}_{\tn{x}} \textcomma{} & \breve{Q}_{\tn{FL}}\textperiod{}
    %     \end{cases}
    % \end{equation}
\end{proposition}

\begin{table*}[t]
    \vspace{5mm}
    \centering%
    \caption{Effect of $L$-bit coding scheme on control performance and design specification}\label{tab:table_metric_computation}
    \input{Tables/table_metric_computation}
\end{table*}

% \begin{proof}
%     See proof of~\cite[Proposition 4]{wolleswinkelQuantifyingControlPerformance2026a}.
% \end{proof}

\input{Sections/arXiv/proof_prop_4}

\section{Illustrative Example}\label{sec:illustrative_example}

Consider the model of a hydro power turbine from~\cite{parkStealthyAdversariesUncertain2019}\FIXME{Removed~\cite{andersonPowerSystemControl2003}}, discretized using~\ac{ZOH} with a frequency of 10 Hz. The matrices $\mtrx{A}$, $\mtrx{B}$, and $\mtrx{B}_{\tn{w}}$ are given by
\begin{equation}
    \mtrx{A} = \begin{bmatrix*}[r]
        0.917 & 0.016 & -0.012 \\
        0.450 & 0.964 &  0.090 \\
        7.560 & 0.069 &  0.550
    \end{bmatrix*} \textcomma{} \quad
    \mtrx{B}_{\tn{w}} = 
    \begin{bmatrix}
        1 & 0 \\
        0 & 1 \\
        0 & 0
    \end{bmatrix} \textcomma{}
\end{equation}
and $\mtrx{B} = \begin{bmatrix} 0 & 0 & 1 \end{bmatrix}^{\T}$. The states $x_{1}$, $x_{2}$, and $x_{3}$ denote the frequency deviation (in Hz), the change in generator output (in Watt), and the change in governor valve position (in rad), respectively. The controller $\sysn{C}$ is given by
\begin{equation}
    \sysn{C}: \qquad\qquad \mtrx{K} = \begin{bmatrix}
        20.498 & 2.092 & 1.529
    \end{bmatrix} \textperiod{} \qquad\quad\strut
\end{equation}
A bounded process noise $\vctr{w}(t)$ acts on the plant, with known bound $\vnorm{\vctr{w}(t)}_{\ell_{\infty}} \leq \ubar{w} = 0.05$. Furthermore, $\mtrx{\Sigma}_{\tn{w}} = 2 \ccdot 10^{\shortminus{} 3} \ccdot \mtrx{I}$. As a design specification, it is imperative that the magnitude of the frequency deviation does not exceed $0.5$ Hz, meaning
\begin{equation}\label{eqn:safe_set_example}
    \setn{X}_{\tn{safe}} = \{ \vctr{x} \in \R^{\dimx} \forwhich{} \vert{} x_{1} \vert{} \leq 0.5 \} \textperiod{}
\end{equation}
The relative weighting of the states is specified as
\begin{equation}
    \mtrx{Q} = \begin{bmatrix*}[r]
        2 & -2 & 0 \\
        -2 & 10 & 0 \\
        0 & 0 & 1
    \end{bmatrix*} \succ 0 \textperiod{}
\end{equation}
We investigate two number formats, namely a \texttt{Q7.8} fixed-point format ($\nbitse = 7$, $\nbitsm = 8$), and a \emph{half-precision} floating-point format ($\nbitse = 5$, $\nbitsm = 10$). We calculate the effect of the $L$-bit coding scheme on control performance, and whether the design specification is maintained. For \texttt{Q7.8}, the resulting values for $\rho_{\tn{FX}}(8, \nbitswm)$ are shown in~\cref{tab:table_metric_computation}. Given $\mtrx{P}$ as in~\eqref{eqn:bmi_ellipsoid}, we can compute the point $\vctr{x} \in \setn{X}_{\tn{reach}}$ with the largest $x_{1}$ magnitude as $\ubar{x}_{1} = \sqrt{[\mtrx{P}^{\shortminus{}1}]_{11}}$, where $[\mtrx{P}^{\shortminus{}1}]_{ij} \in \R$ denotes the entry at the $i$-th row and $j$-th column of $\mtrx{P}^{\shortminus{}1}$. In~\cref{tab:table_metric_computation}, a red underline denotes $\ubar{x}_{1} > 0.5$, implying the design specification~\cref{eqn:safe_set_example} is not met. The values $J(8,\nbitswm)$ for $L = 0, \ldots, 8$ are also shown in~\cref{tab:table_metric_computation}.

%Note that $\vnorm{\vctr{x}(t)} \leq R \implies \frac{1}{\sqrt{\lambda_{\tn{min}}(\mtrx{P}) }} \leq R$.\FIXME{Make this maybe more practical, say the pressure state $x_{2} \leq 6$ at all times, the rated pressure.}
%Furthermore, an illustrative sample of the fixed-point number format limit-cycles, as well as the reachable and minimal volume ellipsoidal invariant set can be seen in~\cref{fig:graph_fixed_point_trajectories}\FIXME{This will only work if we consider a 2D system...}

% \begin{figure}
%     \centering%
%     \input{Graphs/graph_fixed_point_trajectories}
%     \caption{Fixed-point number format effects and reachable and invariant sets}\label{fig:graph_fixed_point_trajectories}
% \end{figure}

Next, we consider the half-precision floating-point format. For the given plant $\sysn{P}$ and controller $\sysn{C}$, we find\FIXME{Here, I removed\cite{kimComputingInducedNorm2020} for space}
% Using the methodology from~\cite{kimComputingInducedNorm2020}\footnotemark{}, we find
%
\begin{equation}
    \vnorm{\mtrx{E}(z)}_{1} = 0.456 \textcomma{} \qquad \vnorm{\mtrx{W}(z)}_{1} = 1.227 \textperiod{}
\end{equation}
Utilizing~\cref{thm:induced_gain}, the resulting metric $\rho_{\tn{FL}}(10, L)$ can be seen in~\cref{tab:table_metric_computation}. Given that $\ubar{w} = 0.05$, we can compute $\ubar{x}_{1} = \rho_{\tn{FL}}(10, L) \ccdot \ubar{w}$, where a red underline in~\cref{tab:table_metric_computation} indicates $\ubar{x}_{1} > 0.5$. Finally, the values $J(10,\nbitswm)$ for $L = 0, \ldots, 10$ are also shown in~\cref{tab:table_metric_computation}.

\subsection{Attack detection and synchronization}\label{sec:attack_detection}

We demonstrate the efficacy of our proposed scheme by means of simulation. Given that the control system operates continuously, we take $\vctr{x}(0) = \vctr{0}$. The process noise $\vctr{w}(t)$ is modeled as a uniform distribution on $[-0.05,0.05]^{2}$. We consider the half-precision floating-point format, and an $L=4$ bit coding scheme. 
% The \ac{HMAC} $H$ is simulated using the \texttt{hmac} package in \Python~\cite{krawczykHMACKeyedHashingMessage1997}.\FIXME{Check out \href{https://docs.python.org/3/library/hmac.html}{Python Doc}}.

We perform a simulation of $T=150$ time steps, and consider two types of attacks. First, consider a replay attack at $t_{\tn{a}} = 20$, given by
\begin{equation}
    \sysn{A}: \qquad\qquad\quad\quad \recieved{\vctr{y}}(t) = \vctr{y}(t - \tau) \textcomma{} \qquad\qquad\qquad\:\:\strut
\end{equation}
with $\tau = 10$. Then, we consider a bias injection attack~\cite{teixeiraSecureControlFramework2015} (a type of~\ac{FDI} attack) at $t_{\tn{a}} = 80$, given by
\begin{equation}
    \sysn{A}: \qquad\qquad \recieved{\vctr{y}}(t) = Q(\beta \ccdot \vctr{y}(t) + (1 - \beta) \ccdot \vctr{x}_{\infty} ) \textcomma{} \quad \strut 
\end{equation}
with $\beta = 0.95 \in (0,1)$ and $\vctr{x}_{\infty} = \vctr{1}$. In line with~\cref{ass:key_secrecy} and~\cref{ass:random_oracle}, we assume the adversary $\sysn{A}$ is aware of the $L$-bit coding scheme, but as $\mf{k}_{\tn{s},i}(t) \notin \setn{I}_{\tn{a}}(t)$, he chooses the last $L$ bits at random, all with equal probability.

% Whenever either attack get detected, we simply stop the attack in this simulation\FIXME{Risky, don't mention}

Lastly, we demonstrate how the proposed authentication scheme handles packet dropouts. From $t \geq 110$ onward, the~\acs{S2C} channel (see~\cref{fig:diagram_networked_control_with_watermarking}) is modeled as
\begin{equation}
    \recieved{\vctr{y}}(t) = \Xi(t) \ccdot \vctr{y}(t) + (1 - \Xi(t)) \ccdot \recieved{\vctr{y}}(t-1) \textcomma{}
\end{equation}
where $\Xi(t)$ follows a Bernoulli distribution with success probability $p = 0.8$, corresponding to a \emph{to-hold} design. We employ a look-ahead window of size $\windowsize = 2$, such that up to two consecutive packet dropouts can be tolerated. 

The simulation results can be seen in~\cref{fig:graph_attacks_synchronization}, where both attack are detected. Under packet dropouts, resynchronization is achieved, until at $t = 131$ three consecutive packets are dropped, causing the detector $\sysn{D}$ to raise a false alarm. Note, however, that the control system is unaffected by this desynchronization.

% For the replay attack at $t = 100$, it becomes evident that the attack is detected and foiled after the very first attack instance. For the~\ac{FDI} attack, the attack is not detected after the first two attack instances, as the adversary happens to have chosen the right digits in the coding scheme. This demonstrates why longer coding schemes with larger $L$ and smaller $\windowsize$ are beneficial, as the former becomes less likely to happen (see~\cref{thm:security_hmac}), but at the cost of loss of control performance and less robustness against network-induced phenomena, respectively. Finally, from $t \geq 300$ the effectiveness of the look-ahead window of size $\windowsize$ is demonstrated. It can be seen that up until $t < 361$ multiple packets are dropped, yet no alarm is raised and synchronization is not lost. However, for $t = 361, \ldots, 365$, four consecutive packets are dropped, and synchronization is lost. As such, a (false) alarm is raised, which is undesirable. Note, however, that the control-loop itself remains functioning normally, demonstrating that the proposed scheme prioritizes availability above all.

\section{Conclusions}\label{sec:conclusions}

In this work, we propose an $L$-bit coding scheme that modifies the~\acp{LSB} of the measurement signals, and analyze the impact of this scheme on control performance. The coding scheme provides message authentication and integrity, whilst prioritizing availability above all else, making it suitable for legacy system with stringent real-time and continuous operation requirements. We also devise a rudimentary yet effective look-ahead window to deal with synchronization issues. Importantly, even under loss of synchronization, the control system remains unaffected.
% Interestingly, for the two respective formats and associated metrics considered, their respective quantities $\ubar{e} = \gamma_{\tn{e}} \approx 2^{-(\nbitsm - L)}$ turned out to be identical.\FIXME{Also mention, that all metrics can be computed independent of the choice of matrix $\mtrx{Q}$} Although not surprising, the control performance loss depend roughly only on the number of bits in the mantissa left.\FIXME{Mention that although worst-case is conservative, for systems which have high performance standards, this is actually vitally important.}

% For future work, we would like to further quantify the impact on control performance not from a worst-case perspective, but rather a distribution-based approach in the same vein as~\cite{vanwingerdenInfluenceFiniteWord1984,millerQuantizerEffectsSteadystate1989,sripadQuantizationErrorsFloatingpoint1978}\FIXME{The last source\cite{sripadQuantizationErrorsFloatingpoint1978} is on signal processing, not on control}, using, e.g.,~\cite[Eq. 12.27]{widrowQuantizationNoiseRoundoff2008} and the~\ac{PQN} model. 
For future work, we would like to test the application of the scheme on a real industrial testbed. Furthermore, inspired by~\cite{ferrariSwitchingMultiplicativeWatermarking2021}, we would also like to investigate more sophisticated synchronization schemes making use of only the measurement channel. Finally, investigating and incorporating packet dropouts directly into the performance metrics would be of interest, as well as considering both output-feedback and dynamic controllers.
% also like to extend and formalize the argument regarding packet dropouts, and its effect on control performance beyond illustrative examples in simulation.

\begin{figure}[htb]
    \centering%
    \input{Graphs/graph_attacks_synchronization}
    \caption{\Acf{MITM} attacks and synchronization}\label{fig:graph_attacks_synchronization}
\end{figure}

%% file: Graphs/graph_mapping_quantization.tex
\begin{tikzpicture}
    % === FIXED, MAP ===
    \begin{axis}[%
        name=fixed_map,
        height=50mm,
        width=50mm,
        xmin = 0,
        xmax = 1,
        ymin = 0.03125,
        ymax = 0.96875,
        title style={%
            at={(0.5,1)},
            font=\footnotesize,
            anchor=south,
            yshift=-2mm},
        title={Fixed-point},
        xmajorticks=false,
        ylabel={$\setn{F}_{\tn{FX}}$},
        ytick=data,
        yticklabel=\empty,
        ytick pos=right,
        tick align=outside,
        ytick style={
            black, 
            opacity=0.8, 
            line width=0.2pt
        },
        ylabel style={xshift=-17pt, yshift=56pt},
        major tick length=3pt
    ]   
        % Add the Q label coordinate
        \coordinate (labelpos) at (axis description cs:0, 0.5);
        % Add levels
        % FIXME: forget plot does not work!
        \addplot[draw=none] table [x=x, y=Q(x), col sep=semicolon] {Data/data_fixed_map_levels.csv};
        % Add x=y axis
        \addplot[forget plot, mark=none, mymediumgray, dotted, domain=0:1] {x};
        % Add mapping
        \addplot[draw=tableau1] table [x=x, y=Q(x), col sep=semicolon] {Data/data_fixed_map_error.csv};
        % Draw the lines to axis
        \draw[mydarkgray] (axis cs:0.425,0.575) -- (axis cs:0.575,0.575) -- (axis cs:0.575,0.425) -- (axis cs:0.425,0.425) -- cycle;
        \coordinate (zoom_fixed_org_ne) at (axis cs:0.575,0.575);
        \coordinate (zoom_fixed_org_sw) at (axis cs:0.425,0.425);
        % Add legend
        % \legend{$\,\setn{X}_{\tn{safe}}$,$\,\mathrm{ker}(\mtrx{C}) \oplus \setn{B}_{\eta}$,$\,\ubar{\setn{X}}_{\tn{reach}}$,$\,\mathrm{ker}(\mtrx{C})$,$\,\vctr{x}_{k}$}
    \end{axis}
    % Add the Q label coordinate
    \node[anchor=east] at (labelpos) {\footnotesize$Q(x)$};
    % === FIXED, ZOOM ===
    \begin{axis}[%
        name=fixed_zoom,
        at={($(fixed_map.south east) + (-2mm,7mm)$)},
        anchor=south east,
        height=25mm,
        width=25mm,
        axis line style=mydarkgray,
        xmin = 0.425,
        xmax = 0.575,
        ymin = 0.425,
        ymax = 0.575,
        xmajorticks=false,
        ymajorticks=false
    ]   
        % Add x=y axis
        \addplot[forget plot, mark=none, mymediumgray, dotted, domain=0:1] {x};
        \addplot[tableau4, ycomb, dashdotted, no marks] table [x=x, y expr=1] {%
            x
            0.46875
            0.53125
        };
        % Add mapping
        \addplot[draw=tableau1] table [x=x, y=Q(x), col sep=semicolon] {Data/data_fixed_map_error.csv};
        % Add corner labels
        \coordinate (zoom_fixed_zoom_ne) at (axis description cs:1, 1);
        \coordinate (zoom_fixed_zoom_sw) at (axis description cs:0, 0);
        % Add stripe labels
        \coordinate (left_bar) at (axis cs:0.46875, 0.425);
        \coordinate (right_bar) at (axis cs:0.53125, 0.425);
        % Add legend
        % \legend{$\,\setn{X}_{\tn{safe}}$,$\,\mathrm{ker}(\mtrx{C}) \oplus \setn{B}_{\eta}$,$\,\ubar{\setn{X}}_{\tn{reach}}$,$\,\mathrm{ker}(\mtrx{C})$,$\,\vctr{x}_{k}$}
    \end{axis}
    % Draw the nodes for connection
    \draw[mydarkgray] (zoom_fixed_org_ne) -- (zoom_fixed_zoom_ne);
    \draw[mydarkgray] (zoom_fixed_org_sw) -- (zoom_fixed_zoom_sw);
    % Add labels for spaced points
    \node[inner sep=1pt] at ($(left_bar) + (-10mm,-4mm)$) (kd) {\color{tableau4}$\scriptscriptstyle{}k \ccdot \delta$};
    \node[inner sep=1pt] at ($(right_bar) + (-5mm,-4mm)$) (kdp) {\color{tableau4}$\scriptscriptstyle{}(k+1) \ccdot \delta$};
    \draw[tableau4] (kd.north) -- (left_bar);
    \draw[tableau4] (kdp.north) -- (right_bar);
    % === FLOATING, MAP ===
    \begin{axis}[%
        name=floating_map,
        at={(fixed_map.north east)},
        anchor=north west,
        xshift=4mm,
        height=50mm,
        width=50mm,
        xmin = 0,
        xmax = 4.35,
        ymin = 0,
        ymax = 4.35,
        title style={%
            at={(0.5,1)},
            font=\footnotesize,
            anchor=south,
            yshift=-2mm},
        title={Floating-point},
        xmajorticks=false,
        ylabel={$\setn{F}_{\tn{FL}}$},
        ytick=data,
        yticklabel=\empty,
        ytick pos=right,
        tick align=outside,
        ytick style={
            black, 
            opacity=0.8, 
            line width=0.2pt
        },
        ylabel style={xshift=-17pt, yshift=56pt},
        major tick length=3pt
    ]   
        % Add levels
        % FIXME: forget plot does not work!
        \addplot[draw=none] table [x=x, y=Q(x), col sep=semicolon] {Data/data_floating_map_levels.csv};
        % Add x=y axis
        \addplot[forget plot, mark=none, mymediumgray, dotted, domain=0:5] {x};
        % Add switching points
        \addplot[tableau4, ycomb, dashdotted, no marks] table [x=x, y expr=5] {Data/data_floating_deltas.csv};
        % Add labels for switching points
        % \node[draw=tableau4, fill=white, inner sep=2pt] at (axis cs:1.5,0.5) {\color{tableau4}$\scriptscriptstyle{}2^{k-1} \ccdot \Delta$};
        \node[draw=tableau4, fill=white, inner sep=2pt] at (axis cs:2,3.8) {\color{tableau4}$\scriptscriptstyle{}2^{k} \ccdot \Delta$};
        \node[draw=tableau4, fill=white, inner sep=2pt] at (axis cs:3.5,1.5) {\color{tableau4}$\scriptscriptstyle{}2^{k + 1} \ccdot \Delta$};
        % Add mapping
        \addplot[draw=tableau1] table [x=x, y=Q(x), col sep=semicolon] {Data/data_floating_map_error.csv};
        % Add legend
        % \legend{$\,\setn{X}_{\tn{safe}}$,$\,\mathrm{ker}(\mtrx{C}) \oplus \setn{B}_{\eta}$,$\,\ubar{\setn{X}}_{\tn{reach}}$,$\,\mathrm{ker}(\mtrx{C})$,$\,\vctr{x}_{k}$}
    \end{axis}
    % === FIXED, ERROR ===
    \begin{axis}[%
        name=fixed_error,
        at={(fixed_map.south west)},
        anchor=north west,
        yshift=-2mm,
        height=25mm,
        width=50mm,
        xmin = 0,
        xmax = 1,
        ymin = 0,
        ymax = 0.052,
        xlabel={$x$},
        ylabel={$\vert{}q(x)\vert{}$},
        xmajorticks=false,
        ymajorticks=false
    ]   
        % Add bound
        \addplot[tableau2, dashed, domain=0:1] {0.03125};
        % Add error
        \addplot[draw=tableau3] table [x=x, y=|e(x)|, col sep=semicolon] {Data/data_fixed_map_error.csv};
        % Add legend
        % \legend{$\,\setn{X}_{\tn{safe}}$,$\,\mathrm{ker}(\mtrx{C}) \oplus \setn{B}_{\eta}$,$\,\ubar{\setn{X}}_{\tn{reach}}$,$\,\mathrm{ker}(\mtrx{C})$,$\,\vctr{x}_{k}$}
        % Add bound
        \node[inner sep=1pt] at (axis cs:0.5,0.041) (e_bar) {\color{tableau2}$\scriptscriptstyle{}\ubar{e}$};
    \end{axis}
    % === FLOATING, ERROR ===
    \begin{axis}[%
        name=floating_error,
        at={(fixed_error.north east)},
        anchor=north west,
        xshift=4mm,
        height=25mm,
        width=50mm,
        xmin = 0,
        xmax = 4.35,
        ymin = 0,
        ymax = 0.36,
        xlabel={$x$},
        xmajorticks=false,
        ymajorticks=false
    ]   
        % Add bound
        \addplot[tableau2, dashed, domain=0:5] {0.0625 * x};
        % Add switching points
        \addplot[tableau4, ycomb, dashdotted, no marks] table [x=x, y expr=5] {Data/data_floating_deltas.csv};
        % Add mapping
        \addplot[draw=tableau3] table [x=x, y=|e(x)|, col sep=semicolon] {Data/data_floating_map_error.csv};
        % Add legend
        % \legend{$\,\setn{X}_{\tn{safe}}$,$\,\mathrm{ker}(\mtrx{C}) \oplus \setn{B}_{\eta}$,$\,\ubar{\setn{X}}_{\tn{reach}}$,$\,\mathrm{ker}(\mtrx{C})$,$\,\vctr{x}_{k}$}
        % Add bound
        \node[inner sep=1pt] at (axis cs:2.7,0.25) (epsilon) {\color{tableau2}$\scriptscriptstyle{}\epsilon \ccdot x$};
    \end{axis}
    % % Draw bounding box
    % \draw[brown] (current bounding box.south west) rectangle (current bounding box.north east);
\end{tikzpicture}

%% file: Diagrams/diagram_networked_control.tex
\begin{tikzpicture}[font=\tikzfont]
    % FROM: https://tex.stackexchange.com/questions/111660/intersection-of-2-lines-not-really-connected-in-tikz
    \tikzset{
        connect/.style args={(#1) to (#2) over (#3) by #4}{
            insert path={
                let \p1=($(#1)-(#3)$), \n1={veclen(\x1,\y1)}, 
                \n2={atan2(\y1,\x1)}, \n3={abs(#4)}, \n4={#4>0 ?180:-180}  in 
                (#1) -- ($(#1)!\n1-\n3!(#3)$) 
                arc (\n2:\n2+\n4:\n3) -- (#2)
            }
        },
    }
    % ------------ ELEMENTS ------------

    % Plant
    \node [draw, fill=white, minimum width=\dynamicswidth, minimum height=\dynamicsheigth] (plant) {$\sysn{P}$};
    \node [above=\headersep of plant] (plant_lab) {Plant};

    % Sensors
    \node [draw, fill=white, minimum width=1.7\dynamicswidth, minimum height=\dynamicsheigth, right=1.5\horizontalblocksep of plant] (sensors) {$\breve{\sysn{S}}\vphantom{\big(}\:\:\qquad\strut$};
    \node [above=\headersep of sensors] (sensors_lab) {Sensors};

    % Coding scheme
    \node[draw=tableau1, fill=tableau7, fill opacity=0.5, minimum width=0.7\dynamicswidth, minimum height=0.7\dynamicsheigth, right=0mm of sensors.center] (coding) {$L$};

    % Splits signals above network
    \coordinate [below left=0.4\verticalsep and 1.5\verticalsependpoints of sensors] (splitUL);
    \coordinate [below right=0.4\verticalsep and 1.5\verticalsependpoints of sensors] (splitUR);

    % Splits signals below network
    \coordinate [below=0.8\networkspacingsep of splitUL] (splitLL);
    \coordinate [below=0.8\networkspacingsep of splitUR] (splitLR);

    % Controller
    \node [draw, minimum width=\dynamicswidth, minimum height=\dynamicsheigth] at ($(splitLL)!0.5!(splitLR) + (0mm,-8mm)$) (controller) {$\sysn{C}$};
    \node [above=\headersep of controller] (controller_lab) {Controller};

    % Detector
    \node [draw, fill=white, minimum width=\dynamicswidth, minimum height=\dynamicsheigth, right=2\horizontalblocksep of controller] (detector) {$\sysn{D}$};

    % Alarm symbol
    \node [draw, fill=white, diamond, minimum size=4mm, above=6mm of detector] (alarm_outer) {};
    \node [draw, fill=white, diamond, minimum size=3mm, label=center:{\tiny \textbf{!}}, inner sep=0mm] (alarm) at (alarm_outer.center) {};

    % ------------ PATHS ------------
    
    % \draw [->, above] (splitUL) |- (plant);
    \draw [->, above] (plant) -- node[above, pos=0.5] {$\vctr{x}(t)$} (sensors);
    \draw [->, right] (sensors) -| node [pos=0.6, right] {$\breve{\vctr{y}}(t) \in \setn{F}^{\dimy}$} (splitUR);
    \draw [->, left] (splitLR) |- node[above, pos=0.7] {$\recieved{\breve{\vctr{y}}}(t)$} (controller);
    \draw [left] (controller) -| node[pos=0.25, above] {$\vctr{u}(t)$} ($(plant.west) + (-1.5\verticalsependpoints, 0mm)$) [->] -- (plant);
    \draw [->] (splitLR |- detector) |- (detector);
    \draw [->] (detector) -- node[right, pos=0.5] {$g(t)$} (alarm_outer);
    
    % Communications network
    \begin{scope}[draw=mymediumgray, text=mymediumgray]
        % Splits network lines
        \coordinate [left=\networklinewidth of splitUL] (splitCAL);
        \coordinate [right=\networklinewidth of splitUL] (splitCAR);
        \coordinate [left=\networklinewidth of splitUL] (splitCSL);
        \coordinate [right=\networklinewidth of splitUR] (splitSCR);
        \coordinate (splitCAL_lower) at (splitLL -| splitCAL);
        \coordinate [right=\networklinewidth of splitLR] (splitCAR_lower);
        
        % Network label
        \coordinate (center_upper_comm) at ($(splitUL)!0.5!(splitUR)$);
        \node [anchor=center, align=center] at ($(splitUL)!0.4!(splitUR) + (0mm, -3mm)$) (actuators_lab) {Communications network};
    
        % Make network lines
        \draw [dashed, thick] (splitCAL) -- (splitCAR);
        \draw [dashed, thick] (splitCSL) -- (splitSCR);
        \draw [dashed, thick] (splitCAL_lower) -- (splitCAR_lower);
    \end{scope}

    % Draw dots at outputs of network
    \draw[fill=black] (splitLR) circle (1pt);

    % ------ ADVERSARY ------

    % Sum
    \node[draw, circle, minimum size=\summationsize] (sum) at ($(splitUR)!0.7!(splitLR)$) {};

    % Adversary
    \node [draw, dashed, fill=white, minimum width=\dynamicswidth, minimum height=\dynamicsheigth, left=1.6\horizontalblocksep of sum] (adversary) {$\sysn{A}$};
    \node [above=\headersep of adversary] (adversary_label) {Adversary};

    % Lines
    \draw[->, mydashdot] (splitUR) -- node[right, pos=0.83] {\scriptsize{}$+$} (sum);
    \draw[->, mydashdot] (adversary) -- node[above, pos=0.45] {$\vctr{a}(t)$} node[above, pos=0.85] {\scriptsize{}$+$} (sum);
    \draw[->, mydashdot] (sum) -- (splitLR);

    % ------ BACKGROUND ------

    \begin{scope}[on background layer]
        \draw[draw=tableau1, fill=tableau7, fill opacity=0.5] ($(detector.north west) + (-2mm,12mm)$) rectangle ($(detector.south east) + (3mm,-2mm)$);
    \end{scope}

    % ------ CHANNELS ------

    % S-C channel
    \node [draw, fill=white, minimum width=3mm, minimum height=2mm, inner sep=1mm, anchor=west, right=0mm of splitSCR] (S-C) {\channelabelsize S2C};

    % Draw bounding box
    % \draw[brown] (current bounding box.south west) rectangle (current bounding box.north east);
\end{tikzpicture}

%% file: Sections/arXiv/proof_prop_1.tex
\begin{proof}
    Given a nonzero $\vctr{a}(t)$ such that $\recieved{\breve{\vctr{y}}}(t) \neq \breve{\vctr{y}}(t)$, due to~\cref{def:hmac}.iii, the digests $\breve{\mf{h}}_{i}(t)$ and $\recieved{\breve{\mf{h}}}_{i}(t)$ are different. \Cref{def:hmac}.ii and~\cref{ass:key_secrecy} ensure that computing the digest $\recieved{\breve{\mf{h}}}_{i}(t)$ for which $\sysn{D}$ will not raise an alarm is infeasible. Then, combining~\cref{ass:random_oracle} and~\cref{ass:key_secrecy}, the best strategy for the adversary $\sysn{A}$ is to pick the altered digest $\recieved{\breve{\mf{h}}}_{i}(t)$ uniformly from $\setn{B}_{L}$. In the best case scenario for the adversary, only a single component $\mf{y}_{i}(t)$ needs to be compromised to launch a successful attack of length $T$. Given that a total of $L$ digits need to be correct, combined with the look-ahead window of size $\windowsize$, this amounts to a binomial distribution with success probability $2^{-L}$ and $\windowsize$ trails. For a successful attack, the former needs to succeed $T$ consecutive times, leading to
    \begin{equation}
        (1 - (1 - 2^{-\nbitswm})^{\windowsize})^{\attacklength} \textcomma{}
    \end{equation}
    which is the probability of a stealthy attack of length $T$.
    % Note that
    % \FIXME{Evidently, under the assumption that H is computationally indistinguishable from uniform, the best the adversary can do, is select the bits from a uniform distribution (or, leave the bits as is, or, put them all to zero). As such, the derived probability arises.}
    % For a false data injection attack, suppose that the modified measurement $\recieved{\vctr{y}}(t) \neq \vctr{y}'(t)$. As a result, some of the digits in the base 2 representation of $\vctr{y}'(t)$ will be different, and due to the properties of the hashing function $H$, this will result in a different $\mf{c}(t)$, which will raise an alarm at the detector. Similarly, for a replay attack, if the adversary tries to transmit $\recieved{\vctr{y}(t)} = \vctr{y}(t - T)$ for some $T > 0$, whilst $\mf{z}(t)$ will be the same, $\mf{key}_{\tn{d}}(t) \neq \mf{key}_{\tn{d}}(t - T)$, and as such, the authentication code $\mf{c}(t - T) \neq H(\mf{z}(t -T) , \mf{key}_{\tn{d}}(t) )$, meaning the attack will be detected at the detector.\FIXME{This proof is super wobbly/shakey}
\end{proof}

%% file: Sections/arXiv/proof_prop_2.tex
\begin{proof}
    From~\cref{eqn:error_two_independent_source}, we can write
    \begin{subequations}\begin{align}
        \vnorm{\vctr{e}(t)}_{\ell_{\infty}} &= \vnorm{ \vctr{x}(t) - (Q(\vctr{x}(t)) + \vctr{d}(t) ) }_{\ell_{\infty}} \\ 
        &\leq \vnorm{\vctr{x}(t) - Q(\vctr{x}(t))}_{\ell_{\infty}} + \vnorm{\vctr{d}(t) }_{\ell_{\infty}} = \ubar{e} \textperiod{}
    \end{align}\end{subequations}
    It it well-known that for uniform quantization $\vert{}x - Q(x)\vert{} \leq 2^{-(\nbitsm + 1)}$~\cite{franklinDigitalControlDynamic2002}. As for the bound $\vnorm{\vctr{d}(t) }_{\ell_{\infty}}$, note that the worst-case error occurs when the last $L$ bits being equal are all flipped (either from all-zero to all-ones, or vice versa). This difference implies $\vnorm{\vctr{d}(t) }_{\ell_{\infty}} \leq 2^{-(\nbitsm - L)} - 2^{-\nbitsm}$, meaning $\ubar{e} = 2^{-(\nbitsm + 1)} + 2^{-(\nbitsm - L)} - 2^{-\nbitsm}$. We can write the closed-loop system as
    % outlined in~\cref{fig:diagram_networked_control_with_watermarking}, as $\vctr{u}(t) = -\mtrx{K} Q'(\vctr{x}(t) ) = - \mtrx{K}(\vctr{x}(t) + \vctr{e}(t))$ from~\cref{eqn:controller}, we can write the closed-loop system as
    %
    \begin{equation}\label{eqn:closed_loop_dynamics_noise_driven}
        \vctr{x}(t + 1) = \mtrx{A}_{\tn{cl}} \vctr{x}(t) + \begin{bmatrix} -\mtrx{B} \mtrx{K} & \mtrx{B}_{\tn{w}} \end{bmatrix} \begin{bmatrix}
            \vctr{e}(t) \\ \vctr{w}(t)
        \end{bmatrix} \textperiod{}
    \end{equation}
    % 
    %which is driven both by the process noise $\vctr{w}(t)$ and the combined quantization and coding error $\vctr{e}(t)$. 
    %
    Importantly, $\vnorm{\vctr{e}(t) }_{\ell_{\infty}} \leq \ubar{e}$ and $\vnorm{\vctr{w}(t) }_{\ell_{\infty}} \leq \ubar{w}$, meaning we can invoke~\cite[Theorem 1]{kafashConstrainingAttackerCapabilities2018}, which gives us the~\ac{BMI} in~\eqref{eqn:bmi_ellipsoid} corresponding to~\cref{eqn:closed_loop_dynamics_noise_driven}. The solution to~\eqref{eqn:bmi_ellipsoid} is the minimum volume ellipsoid $\mtrx{P}$, and the set $\{\vctr{x}(t) \in \R^{\dimx} \forwhich{} \vctr{x}(t)^{\T} \mtrx{P} \vctr{x}(t) \leq 1 \}$. Finally, we note that $\vctr{x}(t) = \mtrx{Q}^{-\frac{1}{2}} \vctr{z}(t)$ from~\cref{eqn:virtual_performance_output}, and as such the scaling factor is equal to
    \begin{equation}
        \mathrm{det}(\mtrx{Q}^{-\frac{1}{2}}) = \textstyle\frac{1}{\mathrm{det}(\mtrx{Q}^{\frac{1}{2}})} = \textstyle\frac{1}{\sqrt{\mathrm{det}(\mtrx{Q})}} \textcomma{}
    \end{equation}
    where the latter follows from $\mtrx{Q} \succ 0$.
\end{proof}

%% file: Sections/arXiv/corollary.tex
In the special noise-free case, meaning $\ubar{w} = 0$, a closed-form expression for $\rho_{\tn{FX}}(\nbitsm, L)$ can be found.

\begin{corollary}
    Consider a fixed-point quantizer $\breve{Q}$ with $\nbitsm$ decimal bits, the $L$-bit coding scheme, and suppose $\ubar{w} = 0$. Then, $\rho_{\tn{FX}}(\nbitsm, L) \propto \ubar{e}$, which implies $\rho_{\tn{FX}}(\nbitsm, L) \approx 2^{-(\nbitsm - L)}$.
\end{corollary}

\begin{proof}
    Note that $\ubar{w} = 0$ means the closed-loop dynamics are given by
    \begin{equation}
        \vctr{x}(t + 1) = \mtrx{A}_{\tn{cl}} \vctr{x}(t)  -\mtrx{B} \mtrx{K} \vctr{e}(t) \textcomma{}
    \end{equation}
    where $\vnorm{\vctr{e}(t)}_{\ell_{\infty}} \leq \ubar{e} $ implies $ \vert{}e_{i}(t) \vert{} \leq \ubar{e}$ for all $i$. As such, all input bounds are equal, and we can leverage~\cite[Remark~1]{kafashConstrainingAttackerCapabilities2018}, where we replace~\cref{eqn:bmi_ellipsoid_bmi} by
    \begin{equation*}
        \begin{bmatrix}
            \alpha \ccdot \hat{\mtrx{P}} - \mtrx{A}_{\tn{cl}}^{\T} \hat{\mtrx{P}} \mtrx{A}_{\tn{cl}} &\mtrx{A}_{\tn{cl}}^{\T} \hat{\mtrx{P}} \mtrx{B}\mtrx{K} \\
            (\mtrx{B} \mtrx{K})^{\T} \hat{\mtrx{P}} \mtrx{A}_{\tn{cl}} & (1 - \alpha) \ccdot \mtrx{I} - (\mtrx{B} \mtrx{K})^{\T} \hat{\mtrx{P}} \mtrx{B} \mtrx{K}
        \end{bmatrix} \succeq 0 \textcomma{}
    \end{equation*}
    with decision variable $\hat{\mtrx{P}} \succ 0$. The minimal volume invariant ellipsoid $\mtrx{E}$ is given by
    \begin{subequations}\label{eqn:vol_ellps_simple_explicit}\begin{align}
        \label{eqn:vol_ellps_simple_explicit_propto_first}
        \rho_{\tn{FX}}(\nbitsm, L) &\propto \textstyle\frac{1}{\sqrt{\mathrm{det}(\mtrx{P})}} = \textstyle\frac{1}{\sqrt{\mathrm{det}(\frac{1}{\ubar{e}^{2}} \ccdot \vphantom{\big(} \smash{\hat{\mtrx{P}}})}} \\
        \label{eqn:vol_ellps_simple_explicit_propto}
        &= \textstyle\frac{1}{ (\frac{1}{\ubar{e}} \ccdot \sqrt{\mathrm{det}(\smash{\hat{\mtrx{P}}})})} \propto \ubar{e} \approx 2^{-(\nbitsm - L)}
    \end{align}\end{subequations}
    where~\cref{eqn:vol_ellps_simple_explicit_propto_first} follows from $\mtrx{Q}$ being independent of the values $\nbitsm$ and $\nbitswm$, and~\cref{eqn:vol_ellps_simple_explicit_propto} follows from $\hat{\mtrx{P}}$ being constant.
\end{proof}

%% file: Sections/arXiv/proof_prop_3.tex
\begin{proof}
    % Consider the floating-point quantizer as discussed in~\cref{sec:quantization} and the~\acp{LSB} authentication scheme as discussed in~\cref{sec:lsb_coding}, we can write the error induces as
    Note that
    \begin{subequations}\begin{align}
        \vnorm{\vctr{e}(t)}_{\ell_{\infty}} &= \vnorm{ \vctr{x}(t) - (Q(\vctr{x}(t)) + \vctr{d}(t) ) }_{\ell_{\infty}} \\ 
        &\leq \vnorm{\vctr{x}(t) - Q(\vctr{x}(t))}_{\ell_{\infty}} + \vnorm{\vctr{d}(t) }_{\ell_{\infty}} \textperiod{}
    \end{align}\end{subequations}
    The error due to quantization $\vnorm{\vctr{x}(t) - Q(\vctr{x}(t))}_{\ell_{\infty}}$ can be upper bounded as a multiplicative error $\vnorm{\vctr{x}(t) - Q(\vctr{x}(t))}_{\ell_{\infty}} \leq \epsilon \ccdot \vnorm{\vctr{x}(t)}_{\ell_{\infty}} + c \approx \epsilon \ccdot \vnorm{\vctr{x}}_{\ell_{\infty}}$\FIXME{removed ~\cite{millerQuantizerEffectsSteadystate1989} for space}, where the constant term is negligible\footnotemark{}~\cite{kontroFloatingpointArithmeticSignal1992}. Let $\nu_{k}$ denote the height of the $k$-th interval\FIXME{This height is not specified...} (see~\cref{fig:graph_mapping_quantization})\FIXME{Nothing to see here....}, which is given by~\cite{nikolicPerformanceAnalysisTwo2025}\FIXME{Also kind of appears in~\cite{sripadQuantizationErrorsFloatingpoint1978}}\FIXME{See if we can remove~\cite{nikolicPerformanceAnalysisTwo2025} for space...}
    \begin{equation}\label{eqn:levels_of_floating}
        \nu_{k} = 2^{k - \bias - \nbitsm} \textperiod{}
    \end{equation}
    Combining~\cref{eqn:levels_of_floating} with the spacing $\Delta = 2^{\nbitsm} \ccdot \smallesstepsize_{\tn{FL}}$~\cite{widrowQuantizationNoiseRoundoff2008}, we find that the rise-over-run $\epsilon$ is given by
    %
    % \begin{subequations}\begin{align}
    %     \epsilon &= \frac{\nu_{k+1} - \nu_{k}}{\Delta_{k+1} - \Delta_{k}} = \frac{2^{k + 1 - \bias- \nbitsm} - 2^{k - \bias- \nbitsm}}{(2 \ccdot \Delta_{k}) - \Delta_{k}} \\
    %     &= \frac{(2^{k + 1} - 2^{k}) \ccdot 2^{- \bias- \nbitsm}}{2^{k} \ccdot \Delta } = \frac{2^{k} \ccdot (2^{1} - 1) \ccdot 2^{- \bias- \nbitsm}}{2^{k} \ccdot \Delta} \\
    %     &= \frac{2^{- \bias- \nbitsm}}{\Delta} = \frac{2^{- \bias}\ccdot 2^{-\nbitsm} }{2^{\nbitsm} \ccdot q} = \frac{2^{- \bias}\ccdot 2^{-\nbitsm} }{2^{\nbitsm} \ccdot (2^{-\nbitsm} \ccdot 2^{-\bias+ 1})} \\
    %     &= \frac{2^{- \bias}\ccdot 2^{-\nbitsm} }{2^{\nbitsm} \ccdot 2^{-\nbitsm} \ccdot 2^{-\bias}\ccdot 2^{1}} 
    %     = \frac{1}{2^{\nbitsm} \ccdot 2^{1}}
    %     = 2^{-(\nbitsm + 1)} \textperiod{}
    % \end{align}\end{subequations}
    \begin{equation}
        \epsilon = \frac{\nu_{k+1} - \nu_{k}}{2^{k+1} \ccdot \Delta - 2^{k} \ccdot \Delta} = \frac{2^{- \bias- \nbitsm}}{\Delta} = 2^{-(\nbitsm + 1)} \textperiod{}
    \end{equation}
    For the coding error $\vctr{d}(t)$, we similarly find $\vnorm{\vctr{d}(t) }_{\ell_{\infty}} \leq \eta \ccdot \vnorm{\vctr{x}(t)}_{\ell_{\infty}}$. The error bound $\eta$ is determined by the first bit which is not a part of the $L$~\acp{LSB}, meaning that $\eta = 2^{-(\nbitsm - L)} - 2^{-\nbitsm}$~\cite{kontroFloatingpointArithmeticSignal1992}. We can write the total error as 
    \begin{equation}\label{eqn:gain_inequality_watermarking}
        \vnorm{\vctr{e}(t)}_{\ell_{\infty}} \leq (\epsilon + \eta) \ccdot \vnorm{\vctr{x}(t)}_{\ell_{\infty}} \textperiod{}
    \end{equation}
    The closed-loop system is given by
    \begin{equation}
        \vctr{x}(t+1) = \mtrx{A}_{\tn{cl}} \vctr{x}(t) - \mtrx{B} \mtrx{K} \mtrx{\Psi}(\vctr{x}(t) ) + \mtrx{B}_{\tn{w}} \vctr{w}(t) \textcomma{}
    \end{equation}
    where $\mtrx{\Psi}: \R^{\dimx} \to \R^{\dimx}$ is a static nonlinearity (see~\cref{fig:diagram_small_gain_watermarking_feedback_interconnection_with_disturbance})\FIXME{Removed~\cite{picassoStabilizationDiscretetimeQuantized2008} for space}. 

    \begin{figure}[htb]
        \centering%
        \input{Diagrams/diagram_small_gain_watermarking_feedback_interconnection_with_disturbance}
        \caption{Error propagation with floating-point}\label{fig:diagram_small_gain_watermarking_feedback_interconnection_with_disturbance}
    \end{figure}

    Noting the the $\ell_{\infty}$-induced norm of an~\acs{LTI} system is its $L_{1}$ norm, we can write~\cref{eqn:gain_inequality_watermarking} as
    \begin{subequations}\label{eqn:float_inf_bound_x}\begin{align}
        \vnorm{\vctr{x}(t)}_{\ell_{\infty}} &\leq \gamma_{E} \ccdot \vnorm{\vctr{e}(t)}_{\ell_{\infty}} +  \gamma_{W} \ccdot \vnorm{\vctr{w}(t)}_{\ell_{\infty}} \\
        &\!\!\!\!\!\leq (\epsilon + \eta) \ccdot \gamma_{E} \ccdot \vnorm{\vctr{x}(t)}_{\ell_{\infty}} +  \gamma_{W} \ccdot \vnorm{\vctr{w}(t)}_{\ell_{\infty}} \textcomma{}
    \end{align}\end{subequations}
    where $\gamma_{E} = \vnorm{\mtrx{E}(z)}_{1}$ and $\gamma_{W} = \vnorm{\mtrx{W}(z)}_{1}$. Here, $\mtrx{E}(z) = -(z \ccdot \mtrx{I} - \mtrx{A}_{\tn{cl}})^{\shortminus{}1} \mtrx{B} \mtrx{K}$ and $\mtrx{W}(z) = (z \ccdot \mtrx{I} - \mtrx{A}_{\tn{cl}})^{\shortminus{}1} \mtrx{B}_{\tn{w}}$ are the transfer function matrices from $\vctr{e}(t)$ to $\vctr{x}(t)$ and $\vctr{w}(t)$ to $\vctr{x}(t)$, respectively. Substituting~\cref{eqn:virtual_performance_output} into~\eqref{eqn:float_inf_bound_x} gives
    \begin{equation}
        \vnorm{\vctr{z}(t)}_{\ell_{\infty}} \leq \frac{\vnorm{\mtrx{Q}}_{\infty} \ccdot \vnorm{\mtrx{W}(z)}_{1}}{1 - (\epsilon + \eta) \ccdot \vnorm{\mtrx{E}(z)}_{1}} \ccdot \vnorm{\vctr{w}(t)}_{\ell_{\infty}} \textperiod{}
    \end{equation}
    Recognizing from~\cref{eqn:metric_floating_gain} that $\vnorm{\mtrx{G}(z)}_{1}$ is the $\ell_{\infty}$-induced gain from $\vctr{w}(t)$ to $\vctr{z}(t)$ proves the result.
\end{proof}

\footnotetext{More precisely, $c = \frac{\smallesstepsize_{\tn{FL}}}{2} \ll \epsilon$. Thus, the upper bound does not hold whenever $x \leq c$.}

%% file: Diagrams/diagram_small_gain_watermarking_feedback_interconnection_with_disturbance.tex
\begin{tikzpicture}[font=\tikzfont]
    % FROM: https://tex.stackexchange.com/questions/111660/intersection-of-2-lines-not-really-connected-in-tikz
    \tikzset{
        connect/.style args={(#1) to (#2) over (#3) by #4}{
            insert path={
                let \p1=($(#1)-(#3)$), \n1={veclen(\x1,\y1)}, 
                \n2={atan2(\y1,\x1)}, \n3={abs(#4)}, \n4={#4>0 ?180:-180}  in 
                (#1) -- ($(#1)!\n1-\n3!(#3)$) 
                arc (\n2:\n2+\n4:\n3) -- (#2)
            }
        },
    }
    % ------------ ELEMENTS ------------

    % Plant
    \node [draw, fill=white, minimum width=6.5\dynamicswidth, minimum height=1.05\dynamicsheigth, inner sep=0pt] (plant) {$\vctr{x}(t+1) = \mtrx{A}_{\tn{cl}} \vctr{x}(t) - \mtrx{B} \mtrx{K} \vctr{e}(t) + \mtrx{B}_{\tn{w}} \vctr{w}(t) $};
    \node [above=\headersep of plant] (plant_lab) {Plant};

    % Splits signals above network
    \coordinate [below left=0.2\verticalsep and 1.5\verticalsependpoints of plant] (splitUL);
    \coordinate [below right=0.2\verticalsep and 1.5\verticalsependpoints of plant] (splitUR);

    % Splits signals below network
    \coordinate [below=0.1\networkspacingsep of splitUL] (splitLL);
    \coordinate [below=0.1\networkspacingsep of splitUR] (splitLR);

    % Controller
    \node [draw, fill=white, minimum width=0.7\dynamicswidth, minimum height=0.7\dynamicsheigth, outer sep=0pt] at ($(splitLL)!0.5!(splitLR) + (0mm,-3mm)$) (controller) {$\nbitsm$};
    \node [above=-1mm of controller.north] (controller_lab) {Quantizer};

    % Sum
    \node [draw, fill=white, circle, minimum size=\summationsize, left=\horizontalblocksep of controller] (sum) {};

    % Watermarker
    \node[draw, fill=white, minimum width=0.7\dynamicswidth, minimum height=0.7\dynamicsheigth, below=5mm of controller, outer sep=0pt] (watermarker) {$L$};
    \node [above=-1mm of watermarker.north] (watermarker_lab) {Coding};]

    % ------------ PATHS ------------
    
    \draw (plant.357) -- ($(plant.357) + (1\verticalsependpoints, 0mm)$) [->] |- node [pos=0.25, right] {$\vctr{x}(t)$} (controller);
    \draw (sum) -| node[pos=0.75, left] {$\vctr{e}(t)$} ($(plant.183) + (-1\verticalsependpoints, 0mm)$) [->] -- (plant.183);
    \draw[->] (controller.west) -- (sum);
    \draw[->] (watermarker) -| (sum);
    \draw[->] ($(controller) + (15mm,0mm)$) |- (watermarker);
    \draw[->] ($(plant.177) + (-10mm,0mm)$) -- node[above] {$\vctr{w}(t)$} (plant.177);
    \draw[->] (plant.3) -- node[above] {$\vctr{z}(t)$} ($(plant.3) + (10mm,0mm)$);

    % ------------ BACKGROUND ------------

    \begin{scope}[on background layer]
        \draw[draw=tableau3, fill=tableau3, fill opacity=0.5, dashed] ($(controller.center) + (-20mm,6mm)$) rectangle ($(controller.center) + (20mm,-12.5mm)$);
        \node[] at ($(controller)!0.5!(watermarker) + (-23mm,0mm)$) {$\color{tableau3}\vctr{\Psi}$};
    \end{scope}
    
\end{tikzpicture}

%% file: Tables/table_metric_computation.tex
\begin{tabular}{cCCCCCCCCCCCC}
    \toprule%
    \multicolumn{2}{c}{$L$} & 0 & 1 & 2 & 3 & 4 & 5 & 6 & 7 & 8 & 9 & 10 \\
    \midrule%
    \multirow{2}{*}{\texttt{Q7.8}}& \rho_{\tn{FX}} & 16.8 \ccdot 10^{\shortminus 3} & 12.6 \ccdot 10^{\shortminus2} & 31.3 \ccdot 10^{\shortminus2} & \coloruline{tableau3}{4.35} & \coloruline{tableau3}{14.1} & \coloruline{tableau3}{35.4} & \coloruline{tableau3}{97.2} & \coloruline{tableau3}{229} &  \coloruline{tableau3}{935} & & \\
     & J & 5.1 \ccdot 10^{\shortminus 5} & 4.58 \ccdot 10^{\shortminus 4} & 2.49 \ccdot 10^{\shortminus 3} & 1.14 \ccdot 10^{\shortminus 2} & 4.89 \ccdot 10^{\shortminus 2}  & 0.20 & 0.82 & 3.31 & 13.27 & & \\[8pt]
     \multirow{2}{*}{\text{Half-precision}}& \rho_{\tn{FL}} & 1.2009 & 1.2026 & 1.2060 & 1.213 & 1.227 & 1.256 & 1.319 & 1.47 & 1.89 & 4.43 & \coloruline{tableau3}{13.3} \\
     & J \ccdot 10^{7} & 95.3 \ccdot 10^{\shortminus2} & 95.3 \ccdot 10^{\shortminus2} & 95.3 \ccdot 10^{\shortminus2} & 95.3 \ccdot 10^{\shortminus2} & 95.4 \ccdot 10^{\shortminus2} & 0.96 & 0.97 & 1.02 & 1.25 & 2.60 & 13.8 \\
    \bottomrule%  
\end{tabular}

%% file: Sections/arXiv/proof_prop_4.tex
\begin{proof}
    According to the closed-loop dynamics given by~\eqref{eqn:closed_loop_dynamics_noise_driven}, the stationary distribution, if it exists, must satisfy
    \begin{subequations}\begin{align}
        \mathbb{E}[\vctr{x}(t)\vctr{x}(t)^{\T}] &= \mathbb{E}[(\mtrx{A}_{\tn{cl}}\vctr{x}(t) + \mtrx{B} \mtrx{K} \mtrx{e}(t) + \mtrx{B}_{\tn{w}} \vctr{w}(t) ) \nonumber\\
        &\!\!\!\!\!\!\!(\mtrx{A}_{\tn{cl}}\vctr{x}(t) + \mtrx{B} \mtrx{K} \mtrx{e}(t) + \mtrx{B}_{\tn{w}} \vctr{w}(t) )^{\T}] \implies \\
        \mtrx{\Sigma}_{\tn{x}} = \mtrx{A}_{\tn{cl}} \mtrx{\Sigma}_{\tn{x}} &\mtrx{A}_{\tn{cl}}^{\T} + \mtrx{B} \mtrx{K} \mtrx{\Sigma}_{\tn{e}} (\mtrx{B} \mtrx{K})^{\T} + \mtrx{B}_{\tn{w}} \mtrx{\Sigma}_{\tn{w}} \mtrx{B}_{\tn{w}}^{\T} \textcomma{}
    \end{align}\end{subequations}
    which is the generalized Lyapunov equation given by~\cref{eqn:generalized_lyaponuv}. For fixed-point, the quantization error and coding error (due to~\cref{ass:random_oracle}) can both be modeled as independent uniformly distributed noise. Combining~\cref{eqn:quantization_error_statistics} and~\cref{eqn:upper_bound_quantization_error}, we have
    \begin{equation}
        \mtrx{\Sigma}_{\tn{e}} \approx \frac{1}{12} \ccdot ( 2^{-2 \ccdot(\nbitsm + 1) } + (2^{-(\nbitsm - L)})^{2} )\ccdot \mtrx{I}  \approx \frac{1}{12} \ccdot \ubar{e}^{2} \ccdot \mtrx{I} \textperiod{}
    \end{equation}
    Similarly, for floating-point,~\cref{ass:random_oracle} and~\cref{eqn:quantization_error_statistics} lead to
    \begin{equation}
        \mtrx{\Sigma}_{\tn{e}} \approx (0.180 \ccdot 2^{-2 \ccdot \nbitsm } + \frac{1}{12} \ccdot 2^{-2\ccdot(\nbitsm - L)}) \ccdot \mtrx{\Sigma}_{\tn{x}} \textperiod{}
    \end{equation}
    Finally, substituting $\vctr{z}(t) = \mtrx{Q}^{-\frac{1}{2}} \vctr{x}(t)$ into~\cref{eqn:distribution_cost}, we get $J(\nbitsm, \nbitswm) = \mathbb{E}[\sum_{t=0}^{\infty} \vctr{x}(t) \mtrx{Q} \vctr{x}(t)^{\T}] = \mathrm{trace}(\mtrx{Q} \mtrx{\Sigma}_{\tn{x}})$.
\end{proof}

%% file: Graphs/graph_attacks_synchronization.tex
\begin{tikzpicture}
    \pgfplotsset{%
        every axis/.append style={%
            yshift=-3mm
        },
        % FROM: ChatGPT 4o
        yticklabel style={/pgf/number format/fixed}, % Keeps y-ticks like 1.0, 1.2
        scaled y ticks=false
        % scaled y ticks=base 10:-3, % Moves 10^3 to the axis label area
        % ytick scale label code/.code={$\ccdot 10^{3}$}, % Formats the top label
    }
    % --- GLOBAL VARS ----
    \pgfmathsetmacro{\attackstartreplay}{20}
    \pgfmathsetmacro{\attackdetectionreplay}{20}
    \pgfmathsetmacro{\attackstartfdi}{21}
    \pgfmathsetmacro{\attackdetectionfdi}{23}
    \newcommand{\axiswave}[1]{\draw ($(#1) + (0mm,0.8mm)$) [bend left=30] to [bend right] (#1) to [bend left=30] ($(#1) + (0mm,-0.8mm)$);}
    % === STATE NORM (replay) ===
    \begin{axis}[%
        name=x_k_norm,
        height=30mm,
        width=40mm,
        yshift=0mm,
        xmin = 0,
        xmax = 33,
        ymin = 0,
        ymax = 0.5,
        xticklabels={\empty},
        axis y line*=left,
        legend style={%
            legend entries={$\vnorm{\vctr{x}(t)}$}
        }
    ]   
        % Add attack instance
        \draw[tableau2, dashdotted, draw opacity=0.5] (axis cs:\attackstartreplay, \pgfkeysvalueof{/pgfplots/ymin}) -- (axis cs:\attackstartreplay, \pgfkeysvalueof{/pgfplots/ymax});
        % Add detection instance
        \draw[tableau5, draw opacity=0.5] (axis cs:\attackdetectionreplay, \pgfkeysvalueof{/pgfplots/ymin}) -- (axis cs:\attackdetectionreplay, \pgfkeysvalueof{/pgfplots/ymax});
        % Add state norm
        \addplot[draw=tableau1] table [x=k, y=x_k_norm, col sep=semicolon] {Data/data_replay_attack_dt.csv};
    \end{axis}
    % Add discontiniuty waves' at end of plot
    \axiswave{x_k_norm.north east}
    \axiswave{x_k_norm.south east}
    % === OUTPUTS (replay) ===
    \begin{axis}[%
        name=y_k,
        at={(x_k_norm.south west)},
        anchor=north west,
        height=34mm,
        width=40mm,
        xmin = 0,
        xmax = 33,
        ymin = -0.4,
        ymax = 0.4,
        xticklabels={\empty},
        axis y line*=left,
        legend style={%
            legend entries={$\recieved{y}_{1}(t)$,$\recieved{y}_{2}(t)$}
        },
        legend columns=2
    ]   
        % Add attack instance
        \draw[tableau2, dashdotted, draw opacity=0.5] (axis cs:\attackstartreplay, \pgfkeysvalueof{/pgfplots/ymin}) -- (axis cs:\attackstartreplay, \pgfkeysvalueof{/pgfplots/ymax});
        % Add detection instance
        \draw[tableau5, draw opacity=0.5] (axis cs:\attackdetectionreplay, \pgfkeysvalueof{/pgfplots/ymin}) -- (axis cs:\attackdetectionreplay, \pgfkeysvalueof{/pgfplots/ymax});
        % Add output y_k_1
        \addplot[draw=tableau2, const plot] table [x=k, y=y_k_1, col sep=semicolon] {Data/data_replay_attack_dt.csv};
        % Add output y_k_2
        \addplot[draw=tableau3, const plot] table [x=k, y=y_k_2, col sep=semicolon] {Data/data_replay_attack_dt.csv};
    \end{axis}
    % Add discontiniuty waves' at end of plot
    \axiswave{y_k.north east}
    \axiswave{y_k.south east}
    % === ATTACK (replay) ===
    \begin{axis}[%
        name=a_k,
        at={(y_k.south west)},
        anchor=north west,
        height=34mm,
        width=40mm,
        xmin = 0,
        xmax = 33,
        ymin = -0.2,
        ymax = 0.6,
        xticklabels={\empty},
        axis y line*=left,
        legend style={%
            legend entries={$a_{1}(t)$,$a_{2}(t)$}
        },
        legend columns=2
    ]   
        % Add attack instance
        \draw[tableau2, dashdotted, draw opacity=0.5] (axis cs:\attackstartreplay, \pgfkeysvalueof{/pgfplots/ymin}) -- (axis cs:\attackstartreplay, \pgfkeysvalueof{/pgfplots/ymax});
        % Add detection instance
        \draw[tableau5, draw opacity=0.5] (axis cs:\attackdetectionreplay, \pgfkeysvalueof{/pgfplots/ymin}) -- (axis cs:\attackdetectionreplay, \pgfkeysvalueof{/pgfplots/ymax});
        % Add output a_y_k_1
        \addplot[draw=tableau4, const plot] table [x=k, y=a_y_k_1, col sep=semicolon] {Data/data_replay_attack_dt.csv};
        % Add output a_y_k_2
        \addplot[draw=tableau5, const plot] table [x=k, y=a_y_k_2, col sep=semicolon] {Data/data_replay_attack_dt.csv};
    \end{axis}
    % Add discontiniuty waves' at end of plot
    \axiswave{a_k.north east}
    \axiswave{a_k.south east}
    % === DETECTION (replay) ===
    \begin{axis}[%
        name=g_k,
        at={(a_k.south west)},
        anchor=north west,
        height=22.3mm,
        width=40mm,
        xmin = 0,
        xmax = 33,
        ymin = -0.2,
        ymax = 1.3,
        ytick={0,1},
        yticklabels={0,1},
        axis y line*=left,
        legend style={%
            legend entries={$g(t)$}
        }
    ]   
        % Add attack instance
        \draw[tableau2, dashdotted, draw opacity=0.5] (axis cs:\attackstartreplay, \pgfkeysvalueof{/pgfplots/ymin}) -- (axis cs:\attackstartreplay, \pgfkeysvalueof{/pgfplots/ymax});
        % Add detection instance
        \draw[tableau5, draw opacity=0.5] (axis cs:\attackdetectionreplay, \pgfkeysvalueof{/pgfplots/ymin}) -- (axis cs:\attackdetectionreplay, \pgfkeysvalueof{/pgfplots/ymax});
        % Add detection signal g_k
        \addplot[draw=tableau8, const plot] table [x=k, y=g_k, col sep=semicolon] {Data/data_replay_attack_dt.csv};
    \end{axis}
    % Add discontiniuty waves' at end of plot
    \axiswave{g_k.north east}
    \axiswave{g_k.south east}
    % === STATE NORM (FDI) ===
    \begin{axis}[%
        name=x_k_norm_fdi,
        at={(x_k_norm.north east)},
        anchor=north west,
        xshift=2mm,
        height=30mm,
        width=37mm,
        yshift=3mm,
        y axis line style={opacity=0},
        ytick=\empty,
        xmin = 5,
        xmax = 35,
        ymin = 0,
        ymax = 1,
        xtick={10,20,30},
        xticklabels={\empty},
        axis y line*=none,
    ]   
        % Add attack instance
        \draw[tableau2, dashdotted, draw opacity=0.5] (axis cs:\attackstartfdi, \pgfkeysvalueof{/pgfplots/ymin}) -- (axis cs:\attackstartfdi, \pgfkeysvalueof{/pgfplots/ymax});
        % Add detection instance
        \draw[tableau5, draw opacity=0.5] (axis cs:\attackdetectionfdi, \pgfkeysvalueof{/pgfplots/ymin}) -- (axis cs:\attackdetectionfdi, \pgfkeysvalueof{/pgfplots/ymax});
        % Add state norm
        \addplot[draw=tableau1] table [x=k, y=x_k_norm, col sep=semicolon] {Data/data_bias_injection_attack_dt.csv};
    \end{axis}
    % Add discontiniuty waves' at end of plot
    \axiswave{x_k_norm_fdi.north east}
    \axiswave{x_k_norm_fdi.south east}
    \axiswave{x_k_norm_fdi.north west}
    \axiswave{x_k_norm_fdi.south west}
    % === OUTPUTS (FDI) ===
    \begin{axis}[%
        name=y_k_fdi,
        at={(x_k_norm_fdi.south west)},
        anchor=north west,
        height=34mm,
        width=37mm,
        y axis line style={opacity=0},
        ytick=\empty,
        xmin = 5,
        xmax = 35,
        ymin = -0.7,
        ymax = 0.7,
        xtick={10,20,30},
        xticklabels={\empty},
        axis y line*=none
    ]   
        % Add attack instance
        \draw[tableau2, dashdotted, draw opacity=0.5] (axis cs:\attackstartfdi, \pgfkeysvalueof{/pgfplots/ymin}) -- (axis cs:\attackstartfdi, \pgfkeysvalueof{/pgfplots/ymax});
        % Add detection instance
        \draw[tableau5, draw opacity=0.5] (axis cs:\attackdetectionfdi, \pgfkeysvalueof{/pgfplots/ymin}) -- (axis cs:\attackdetectionfdi, \pgfkeysvalueof{/pgfplots/ymax});
        % Add output y_k_1
        \addplot[draw=tableau2, const plot] table [x=k, y=y_k_1, col sep=semicolon] {Data/data_bias_injection_attack_dt.csv};
        % Add output y_k_2
        \addplot[draw=tableau3, const plot] table [x=k, y=y_k_2, col sep=semicolon] {Data/data_bias_injection_attack_dt.csv};
    \end{axis}
    % Add discontiniuty waves' at end of plot
    \axiswave{y_k_fdi.north east}
    \axiswave{y_k_fdi.south east}
    \axiswave{y_k_fdi.north west}
    \axiswave{y_k_fdi.south west}
    % === ATTACK (FDI) ===
    \begin{axis}[%
        name=a_k_fdi,
        at={(y_k_fdi.south west)},
        anchor=north west,
        height=34mm,
        width=37mm,
        y axis line style={opacity=0},
        ytick=\empty,
        xmin = 5,
        xmax = 35,
        ymin = -0.2,
        ymax = 0.6,
        xtick={10,20,30},
        xticklabels={\empty},
        axis y line*=none
    ]   
        % Add attack instance
        \draw[tableau2, dashdotted, draw opacity=0.5] (axis cs:\attackstartfdi, \pgfkeysvalueof{/pgfplots/ymin}) -- (axis cs:\attackstartfdi, \pgfkeysvalueof{/pgfplots/ymax});
        % Add detection instance
        \draw[tableau5, draw opacity=0.5] (axis cs:\attackdetectionfdi, \pgfkeysvalueof{/pgfplots/ymin}) -- (axis cs:\attackdetectionfdi, \pgfkeysvalueof{/pgfplots/ymax});
        % Add output a_y_k_1
        \addplot[draw=tableau4, const plot] table [x=k, y=a_y_k_1, col sep=semicolon] {Data/data_bias_injection_attack_dt.csv};
        % Add output a_y_k_2
        \addplot[draw=tableau5, const plot] table [x=k, y=a_y_k_2, col sep=semicolon] {Data/data_bias_injection_attack_dt.csv};
    \end{axis}
    % Add discontiniuty waves' at end of plot
    \axiswave{a_k_fdi.north east}
    \axiswave{a_k_fdi.south east}
    \axiswave{a_k_fdi.north west}
    \axiswave{a_k_fdi.south west}
    % === DETECTION (FDI) ===
    \begin{axis}[%
        name=g_k_fdi,
        at={(a_k_fdi.south west)},
        anchor=north west,
        height=22.3mm,
        width=37mm,
        y axis line style={opacity=0},
        ytick=\empty,
        xmin = 5,
        xmax = 35,
        ymin = -0.2,
        ymax = 1.3,
        xlabel={$t$},
        xtick={11,21,31},
        xticklabels={70,80,90},
        axis y line*=none
    ]   
        % Add attack instance
        \draw[tableau2, dashdotted, draw opacity=0.5] (axis cs:\attackstartfdi, \pgfkeysvalueof{/pgfplots/ymin}) -- (axis cs:\attackstartfdi, \pgfkeysvalueof{/pgfplots/ymax});
        % Add detection instance
        \draw[tableau5, draw opacity=0.5] (axis cs:\attackdetectionfdi, \pgfkeysvalueof{/pgfplots/ymin}) -- (axis cs:\attackdetectionfdi, \pgfkeysvalueof{/pgfplots/ymax});
        % Add detection signal g_k
        \addplot[draw=tableau8, const plot] table [x=k, y=g_k, col sep=semicolon] {Data/data_bias_injection_attack_dt.csv};
    \end{axis}
    % Add discontiniuty waves' at end of plot
    \axiswave{g_k_fdi.north east}
    \axiswave{g_k_fdi.south east}
    \axiswave{g_k_fdi.north west}
    \axiswave{g_k_fdi.south west}
    % === STATE NORM (SYNC) ===
    \begin{axis}[%
        name=x_k_norm_sync,
        at={(x_k_norm_fdi.north east)},
        anchor=north west,
        xshift=2mm,
        height=30mm,
        width=40mm,
        yshift=3mm,
        ytick=\empty,
        xmin = 28,
        xmax = 58,
        ymin = 0,
        ymax = 0.6,
        xtick={10,20,30},
        xticklabels={\empty},
        axis y line*=right,
        legend style={%
            legend entries={$\Xi(t)$},
            at={($(1,1) + (-1mm,-1mm)$)},
            anchor=north east
        }
    ]   
        % Add package drops
        \addplot[draw=tableau4, ycomb, dashed, forget plot] table [x expr={\thisrow{xi_k}==1?\thisrow{k}:nan}, y expr=2, col sep=semicolon] {Data/data_synchronization_dt.csv};
        % Add state norm
        \addplot[draw=tableau1, forget plot] table [x=k, y=x_k_norm, col sep=semicolon] {Data/data_synchronization_dt.csv};
        % For legend entry
        \addplot[draw=tableau4, dashed] {-1};
    \end{axis}
    % Add discontiniuty waves' at end of plot
    \axiswave{x_k_norm_sync.north west}
    \axiswave{x_k_norm_sync.south west}
    % === OUTPUTS (SYNC) ===
    \begin{axis}[%
        name=y_k_sync,
        at={(x_k_norm_sync.south west)},
        anchor=north west,
        height=34mm,
        width=40mm,
        ytick=\empty,
        xmin = 28,
        xmax = 58,
        ymin = -0.4,
        ymax = 0.4,
        xtick={10,20,30},
        xticklabels={\empty},
        axis y line*=right
    ]   
        % Add output y_k_1
        \addplot[draw=tableau2, const plot] table [x=k, y=y_k_1, col sep=semicolon] {Data/data_synchronization_dt.csv};
        % Add output y_k_2
        \addplot[draw=tableau3, const plot] table [x=k, y=y_k_2, col sep=semicolon] {Data/data_synchronization_dt.csv};
    \end{axis}
    % Add discontiniuty waves' at end of plot
    \axiswave{y_k_sync.north west}
    \axiswave{y_k_sync.south west}
    % === ATTACK (SYNC) ===
    \begin{axis}[%
        name=a_k_sync,
        at={(y_k_sync.south west)},
        anchor=north west,
        height=34mm,
        width=40mm,
        ytick=\empty,
        xmin = 28,
        xmax = 58,
        ymin = -0.2,
        ymax = 0.6,
        xtick={10,20,30},
        xticklabels={\empty},
        axis y line*=right
    ]   
        % Add output a_y_k_1
        \addplot[draw=tableau4, const plot, domain=0:100] {0};
        % Add output a_y_k_2
        \addplot[draw=tableau5, const plot, domain=0:100] {0};
    \end{axis}
    % Add discontiniuty waves' at end of plot
    \axiswave{a_k_sync.north west}
    \axiswave{a_k_sync.south west}
    % === DETECTION (SYNC) ===
    \begin{axis}[%
        name=g_k_sync,
        at={(a_k_sync.south west)},
        anchor=north west,
        height=22.3mm,
        width=40mm,
        ytick=\empty,
        xmin = 28,
        xmax = 58,
        ymin = -0.2,
        ymax = 1.3,
        xtick={30,40,50},
        xticklabels={120,130,140},
        axis y line*=right
    ]   
        % Add package drops
        \addplot[draw=tableau4, ycomb, dashed] table [x expr={\thisrow{xi_k}==1?\thisrow{k}:nan}, y expr=2, col sep=semicolon] {Data/data_synchronization_dt.csv};
        % Add detection signal g_k
        \addplot[draw=tableau8, const plot] table [x=k, y=g_k, col sep=semicolon] {Data/data_synchronization_dt.csv};
    \end{axis}
    % Add discontiniuty waves' at end of plot
    \axiswave{g_k_sync.north west}
    \axiswave{g_k_sync.south west}
    % % Draw bounding box
    % \draw[brown] (current bounding box.south west) rectangle (current bounding box.north east);
\end{tikzpicture}

%% file: Sections/acronyms.tex
\begin{acronym}[\hspace{0.8in}] % 0.8in is also used by the nomenclature
    \acro{A/D}{\ul{a}nalog-to-\ul{d}igital}
    \acro{ACK}{\ul{ack}nowledgement}
    \acro{ARE}{\ul{a}lgebraic \ul{R}iccati \ul{e}quation}%
    \acro{BMI}{\ul{b}ilinear \ul{m}atrix \ul{i}nequality}%
    \acrodefplural{BMI}{\ul{b}ilinear \ul{m}atrix \ul{i}nequalities}%
    \acro{BITW}{\ul{b}ump-\ul{i}n-\ul{t}he-\ul{w}ire}%
    \acro{C2A}{\ul{c}ontroller-\ul{to}-\ul{a}ctuators}%
    \acro{C2S}{\ul{c}ontroller-\ul{to}-\ul{s}ensors}%
    \acro{CETC}{\ul{c}ontinuous \ul{e}vent-\ul{t}riggered \ul{c}ontrol}%
    \acro{CIA}{\ul{c}onfidentiality, \ul{i}ntegrity, \ul{a}vailability}%
    \acro{CPS}{\ul{c}yber-\ul{p}hysical \ul{s}ystem}%
    \acro{DARE}{\ul{d}iscrete \ul{a}lgebraic \ul{R}iccati \ul{e}quation}%
    \acro{DC}{\ul{d}irect \ul{c}urrent}%
    \acro{DoS}{\ul{d}enial-\ul{o}f-\ul{s}ervice}%
    \acro{DTU}{\ul{D}anmarks \ul{T}ekniske \ul{U}niversitet}%
    \acro{EISS}{\ul{e}xponentially \ul{i}nput-to-\ul{s}tate \ul{s}table}%
    % FIXME: Define pluralities such as S: stable P: stability
    \acro{ETC}{\ul{e}vent-\ul{t}riggered \ul{c}ontrol}%
    \acro{FDI}{\ul{f}alse \ul{d}ata \ul{i}njection}%
    \acro{FWL}{\ul{f}inite \ul{w}ord \ul{l}ength}%
    \acro{FWT}{\ul{f}loating \ul{w}ind \ul{t}urbine}%
    \acro{FIR}{\ul{f}inite \ul{i}mpulse \ul{r}esponse}%
    \acro{FPGA}{\ul{f}ield-\ul{p}rogrammable \ul{g}ate \ul{a}rray}%
    \acro{GES}{\ul{g}lobally \ul{e}xponentially \ul{s}table}%
    \acrodefplural{GES}[GES]{\ul{g}lobal \ul{e}xponential \ul{s}tability}%
    \acro{GUB}{\ul{g}lobally \ul{u}ltimately \ul{b}ounded}%
    \acro{HMAC}{\ul{h}ash-based \ul{m}essage \ul{a}uthentication \ul{c}ode}%
    \acro{ICS}{\ul{i}ndustrial \ul{c}ontrol \ul{s}ystem}%
    \acro{IET}{\ul{i}nter-\ul{e}vent \ul{t}ime}%
    \acro{IDS}{\ul{i}ntrusion \ul{d}etection \ul{s}ystem}%
    \acro{i.i.d.}{\ul{i}ndependent and \ul{i}dentically \ul{d}istributed}%
    \acro{iff}{\ul{if} and only i\ul{f}}%
    \acro{IoT}{\ul{i}nternet \ul{o}f \ul{t}hings}%
    \acro{ISS}{\ul{i}nput-to-\ul{s}tate \ul{s}table}%
    \acrodefplural{ISS}[ISS]{\ul{i}nput-to-\ul{s}tate \ul{s}tability}
    \acro{IT}{\ul{i}nformation \ul{t}echnology}%
    \acro{KDF}{\ul{k}ey \ul{d}erivation \ul{f}unction}%
    \acro{KKT}{\ul{K}arush–\ul{K}uhn–\ul{T}ucker}%
    \acro{KL}{\ul{K}ullback-\ul{L}eibler}%
    \acro{LE}{\ul{l}ittle-\ul{e}ndian}
    \acro{LI}{\ul{l}inear \ul{i}mpulsive}%
    \acro{LLS}{\ul{l}inear \ul{l}east \ul{s}quares}%
    \acro{LMI}{\ul{l}inear \ul{m}atrix \ul{i}nequality}%
    \acrodefplural{LMI}{\ul{l}inear \ul{m}atrix \ul{i}nequalities}%
    \acro{LSB}{\ul{l}east \ul{s}ignificant \ul{b}it}%
    \acro{LQG}{\ul{l}inear–\ul{q}uadratic–\ul{G}aussian}%
    \acro{LQR}{\ul{l}inear–\ul{q}uadratic–\ul{r}egulator}%
    \acro{LTI}{\ul{l}inear \ul{t}ime-\ul{i}nvariant}%
    \acro{MACE}{\ul{m}inimum-\ul{a}verage-\ul{c}ycle-\ul{e}quivelant}%
    \acro{MILP}{\ul{m}ixed-\ul{i}nteger \ul{l}inear \ul{p}rogramming}%
    \acro{MIMO}{\ul{m}ultiple-\ul{i}nput and \ul{m}ultiple-\ul{o}utput}%
    \acro{MITM}{\ul{m}an-\ul{i}n-\ul{t}he-\ul{m}iddle}%
    \acro{MSS}{\ul{m}ean \ul{s}quare \ul{s}table}%
    \acrodefplural{MSS}[MSS]{\ul{m}ean \ul{s}quare \ul{s}tability}%
    \acro{MPC}{\ul{m}odel \ul{p}redictive \ul{c}ontrol}%
    \acro{NaN}{\ul{n}ot-\ul{a}-\ul{n}umber}%
    \acro{NCPS}{\ul{n}etworked \ul{c}yber-\ul{p}hysical \ul{s}ystem}%
    \acro{NCS}{\ul{n}etworked \ul{c}ontrol \ul{s}ystem}%
    \acro{NMP}{\ul{n}on-\ul{m}inimum \ul{p}hase}%
    \acro{OT}{\ul{o}perational \ul{t}echnology}%
    \acro{PD}{\ul{p}ositive-\ul{d}efinite}%
    \acro{PDF}{\ul{p}robability \ul{d}ensity \ul{f}unction}%
    \acro{PETC}{\ul{p}eriodic \ul{e}vent-\ul{t}riggered \ul{c}ontrol}%
    \acro{PID}{\ul{p}roportional–\ul{i}ntegral–\ul{d}erivative}%
    \acro{PLC}{programmable logic controller}%
    \acro{PMF}{\ul{p}robability \ul{m}ass \ul{f}unction}%
    \acro{PRF}{\ul{p}seudo\ul{r}andom \ul{f}unction}%
    \acro{PSD}{\ul{p}ositive \ul{s}emi-\ul{d}efinite}%
    \acro{PWL}{\ul{p}iece\ul{w}ise \ul{l}inear}%
    \acro{PQN}{\ul{p}seudorandom \ul{q}uantization \ul{n}oise}%
    \acro{QP}{\ul{q}uadratic \ul{p}rogram}%
    \acro{QTP}{\ul{q}uadruple-\ul{t}ank \ul{p}rocess}%
    \acro{RTU}{\ul{r}emote \ul{t}erminal \ul{u}nit}%
    \acro{S2C}{\ul{s}ensors-\ul{to}-\ul{c}ontroller}%
    \acro{SCADA}{\ul{s}upervisory \ul{c}ontrol \ul{a}nd \ul{d}ata \ul{a}cquisition}%
    \acro{SD}{\ul{s}tandard \ul{d}eviation}%
    \acro{s.e.}{\ul{s}tandard \ul{e}rror}%
    \acro{SOCP}{\ul{s}econd \ul{o}rder \ul{c}one \ul{p}rogramming}%
    \acro{RHS}{\ul{r}ight-\ul{h}and \ul{s}ide}%
    \acro{SISO}{\ul{s}ingle-\ul{i}nput and \ul{s}ingle-\ul{o}utput}%
    \acro{SL}{\ul{s}witched \ul{l}inear}%
    \acro{STC}{\ul{s}elf-\ul{t}riggered \ul{c}ontrol}%
    \acro{TTC}{\ul{t}ime-\ul{t}riggered \ul{c}ontrol}%
    \acro{VHMPC}{\ul{v}ariable \ul{h}orizon \ul{m}odel \ul{p}redictive \ul{c}ontrol}%
    \acro{WM}{\ul{w}ater\ul{m}arking}%
    \acro{ZDA}{\ul{z}ero \ul{d}ynamics \ul{a}ttack}%
    \acro{ZOH}{{\ul{z}}ero-\ul{o}rder \ul{h}old}%
\end{acronym}%